\documentclass[12pt]{article}
\usepackage{amsmath,amssymb,amsthm,amsbsy,amsfonts,relsize,mathrsfs,cases,ulem,bm,latexsym,bbm,dsfont,mathtools,extarrows,tipa,yhmath,abraces,ragged2e,setspace,color,colortbl,float,caption,booktabs,multirow,fancyhdr,changepage,titlesec,longtable,booktabs,graphicx,threeparttable,rotating,enumitem,geometry}

\usepackage{natbib}
\usepackage{url} 
\usepackage[bookmarks=true]{hyperref}
\hypersetup{hidelinks}
\usepackage{cleveref}

\makeatletter
\newcommand{\labeltext}[1]{%
	#1%
	\protected@edef\@currentlabel{#1}%
}
\makeatother
\usepackage{latexsym,longtable,booktabs,graphicx,rotating,multirow,caption,float,colortbl,threeparttable,bigstrut,array}
\usepackage{color}
\usepackage{colortbl}
\usepackage[dvipsnames]{xcolor}
\usepackage[most]{tcolorbox}
\usepackage{authblk}
\usepackage{setspace}
\usepackage{tikz}
\usepackage{enumitem}
\usepackage{footmisc}
\makeatletter

\renewcommand{\maketag@@@}[1]{\hbox{\m@th\normalsize\normalfont#1}}%

\makeatother
\renewcommand{\theequation}{\arabic{section}.\arabic{equation}}
\makeatletter
\@addtoreset{equation}{section}
\makeatother
\makeatletter
\def\widebreve{\mathpalette\wide@breve}
\def\wide@breve#1#2{\sbox\z@{$#1#2$}%
	\mathop{\vbox{\m@th\ialign{##\crcr
				\kern0.08em\brevefill#1{0.8\wd\z@}\crcr\noalign{\nointerlineskip}%
				$\hss#1#2\hss$\crcr}}}\limits}
\def\brevefill#1#2{$\m@th\sbox\tw@{$#1($}%
	\hss\resizebox{#2}{\wd\tw@}{\rotatebox[origin=c]{90}{\upshape(}}\hss$}
\makeatletter
\makeatletter
\newcommand{\mathleft}{\@fleqntrue\@mathmargin0pt}
\newcommand{\mathcenter}{\@fleqnfalse}
\makeatother
\theoremstyle{definition}
\newtheorem{assumption}{Assumption}[section]
\newtheorem{theorem}{Theorem}[section]

\newtheorem{theorem*}{Theorem}[section]

\newtheorem{definition}{Definition}

\newtheorem{lemma}{Lemma}[section]

\renewenvironment{proof}{{\noindent\it (Proof).}}{\hfill$\blacksquare$\\\par}
\newcommand{\invn}{\frac{1}{n}}

\newcommand{\Jc}{\dgamma}
\newcommand{\cU}{\mathcal{U}}
\newcommand{\cV}{\mathcal{V}}
\newcommand{\Ta}{T_{1}}
\newcommand{\Tb}{T_{2}}
\newcommand{\Tc}{T_{3}}

\newcommand{\stl}{\setlength{\tabcolsep}}
\newcommand{\stla}{\setlength{\abovecaptionskip}}

\newcommand{\rx}{\mathrm{x}}
\newcommand{\suml}{\sum\limits}

\newcommand{\fancya}{\mathscr{A}}
\newcommand{\fancys}{\mathscr{S}}
\newcommand{\fancyk}{\mathscr{K}}

\newcommand{\Ka}{K_{1}}
\newcommand{\Kb}{K_{2}}
\newcommand{\Kc}{K_{3}}
\newcommand{\PKa}{P_{\Ka}}
\newcommand{\PKb}{P_{\Kb}}
\newcommand{\PKc}{P_{\Kc}}

\newcommand{\cM}{\mathcal{M}}
\newcommand{\E}{\mathrm{E}}

\newcommand{\da}{d_{\alpha}}

\newcommand{\cD}{\mathscr{D}}

\newcommand{\dxi}{d_{\xi}}
\newcommand{\norm}{\big\Vert}
\newcommand{\dgamma}{d_{\gamma}}
\newcommand{\dl}{d_{\mu}}
\newcommand{\dmu}{\dl}
\newcommand{\dtau}{d_{\tau}}

\DeclareMathAlphabet{\mathpzc}{OT1}{pzc}{m}{it}

\newcommand \ninf {n\rightarrow\infty}
\newcommand \irangen {i=1,\ldots,n}
\newcommand \calX {\mathcal{X}}

\newcommand \calZ {\mathcal{Z}}

\newcommand \calP {\mathcal{P}}

\newcommand{\dvartheta}{d_{\vartheta}}
\newcommand{\dtheta}{d_\theta}

\newcommand{\rdelta}{r_{\delta}}
\newcommand{\rG}{r_{G}}
\newcommand{\rlambda}{r_{\lambda}}
\newcommand{\rw}{r_{G\delta}}
\newcommand{\ry}{r_{\lambda\delta}}
\newcommand{\dlambda}{d_{\lambda}}
\newcommand{\dbeta}{d_{\beta}}
\newcommand{\RNum}[1]{\uppercase\expandafter{\romannumeral #1\relax}}
\newcommand{\titlename}{Wald inference on varying coefficients}
\begin{document}
		\title{\bf \titlename
        \thanks{Abhimanyu Gupta's research was supported by the Leverhulme Trust via grant RPG-2024-038. Xi Qu's research was supported by the National Natural Science Foundation of China via grants 72595872 and 72222007. Jiajun Zhang's research was supported by the National Natural Science Foundation of China via grant 72503134.}}
\author[1]{Abhimanyu Gupta\thanks{Email: abhimanyu.g@queensu.ca}}
\author[2]{Xi Qu\thanks{Email: xiqu@sjtu.edu.cn.}}
\author[3]{Sorawoot Srisuma\thanks{Email: s.srisuma@nus.edu.sg.}}
\author[4]{Jiajun Zhang\thanks{Email: jiajun30@suibe.edu.cn.}}

\affil[1]{Department of Economics, Queen's University, Dunning Hall, 94 University Avenue, Kingston, Ontario K7L 3N6, Canada, and Department of Economics, University of Essex, Wivenhoe Park, Colchester, CO4 3SQ, UK.}
\affil[2]{Department of Economics, Antai College of Economics and Management, Shanghai Jiao Tong University, 1954 Huashan Road, Shanghai, 200030, China PRC. }
\affil[3]{Department of Economics, National University of Singapore, 1 Arts Link, 117570, Singapore.}
\affil[4]{International Business School,
Shanghai University of International Business and Economics, Shanghai, 201620, China PRC.}

\renewcommand*{\Affilfont}{\small\it} 
		\maketitle	
	\bigskip

	\begin{abstract}
We present simple to implement Wald-type statistics that deliver a general nonparametric inference theory for linear restrictions on varying coefficients in a range of regression models allowing for cross-sectional or spatial dependence. We provide a general central limit theorem that covers a broad range of error spatial dependence structures, allows for a degree of misspecification robustness via nonparametric spatial weights and permits inference on both varying regression and spatial dependence parameters. Using our method, we first uncover evidence of constant returns to scale in the Chinese nonmetal mineral industry’s production function, and then show that Boston house prices respond nonlinearly to proximity to employment centers. A simulation study confirms that our tests perform very well in finite samples.
\end{abstract}	
\noindent%
\vfill

\spacing{1.8} 

\section{Introduction}\label{sec:intro}
This paper develops a general framework for inference on varying coefficient regression models that is as easy to implement in practice as familiar Wald tests for finite-dimensional parameters. We also allow for spatial autoregressive (SAR) structure to permit interaction across the cross-section of economic agents. SAR models are a popular and parsimonious tool for researchers working in settings where economic agents interact via economic or social links. Such links can be geographic proximity, or proximity in some more general sense such as a social network. The SAR model can accommodate these types of general links easily, and this goes some way to explain its appeal, see e.g. \cite{Helmers2014} and \cite{hsieh2018} for examples of applications to peer effects. No surprise then that the baseline SAR model, due to \citet{cliff1981spatial}, has become much studied by econometricians. This model has long been popular with regional scientists, see \citet{anselin1988spatial} for an early book length treatment; important econometric contributions that established the foundation for rigorous econometric analysis of such models include \citet{kelejian1998generalized}, \citet{lee2004asymptotic} and \citet{Kuersteiner2020}. 

In keeping with the growing importance of SAR models and their ability to parsimoniously control for cross-sectional dependence in multiple regression designs, various strands of the literature have explored how to widen their scope. One such is the focus on varying coefficient SAR models, wherein model parameters are allowed to be unknown functions of some observed economic variable(s). This permits a flexible effect of both the regressors and spatial dependence on the outcome variable, and indeed varying coefficient models are abundant in statistics and econometrics (see e.g. \citet{Fan1999,Fan2008}). In the SAR context, varying coefficient models have been studied by \citet{Sun2014} and a sequence of important papers by \citet{Sun2016}, \citet{malikov2017semiparametric}, \citet{Sun2018} and \citet{Sun2024semiparametric}, for example. More generally, semiparametric SAR models have attracted much interest, see e.g. \citet{Su2010}, \citet{Su2012}, \citet{Robinson2012c} and \citet{Zhang2013}, but their focus tends to be on inference on the parametric component.

In this paper, we show how inference on varying coefficients in a range of SAR models can be conducted just like familiar parametric Wald tests for linear restrictions. Standard varying coefficient regression without cross-sectional dependence is naturally covered as a special case. The key idea is that if series approximations are used for the nonparametric varying coefficients then the series coefficients can be employed for very easy practical inference. This is motivated by the scientific aim of simplification of nonparametric/semiparametric inference. Indeed, a growing literature stresses such ideas, see e.g \citet{Ackerberg2012}, \citet{Gupta2018c} and \citet{Korolev2019}, amongst others. Naturally our methods are also applicable to the usual varying coefficient regression model with no cross-sectional dependence, see e.g. \cite{Ahmad2005}, which is a simply a particular case of our theory. 

To explain the intuition of our approach, suppose a sample of $n$ observations is available. With our method, the researcher can reduce  the econometric problem of inference on an infinite-dimensional object to inference on a growing number of series coefficients, say $p$. In practice, this simply means writing down the usual Wald test statistic $\mathscr{W}$ for $p$ linear restrictions, observing that $p\rightarrow\infty$ as $n\rightarrow\infty$ if it is the length of a series approximation, and appealing to a theorem that establishes $(\mathscr{W}-p)/\sqrt{2p}\overset{d}{\rightarrow} N(0,1)$ as $n\rightarrow\infty$, under the null hypothesis. The idea stems from the fact that as $p\rightarrow\infty$, a $(\chi^2_p-p)/\sqrt{2p}$ random variable approaches a standard normal variate, see e.g. \citet{DeJong1994}, \citet{Hong1995}, \citet{Gupta2018c} and \citet{Gupta2023}.

Using this series expansion approach, we theoretically justify Wald-type test statistics for inference on varying coefficients. We first cover a baseline high-order SAR model with varying regression coefficients but constant spatial lag parameters and extend the model to allow for spatial error dependence. The baseline model is similar to the one estimated by \citet{Sun2014}, but also features a growing number of spatial lags and covariates. We then show how to allow for a degree of misspecification robustness by incorporating nonparametric spatial weights \textit{\`a la} \citet{pinkse2002}. Our approach is also shown to work for models with varying coefficients on the spatial lags, in the spirit of \citet{malikov2017semiparametric}.

We then apply our methods in two different settings where varying coefficients
have been shown to play an important role. In the first application, we test
for constant returns to scale (CRS) in the production function of the
Chinese nonmetal mineral industry. This re-visits the empirical study in \citet{li2002semiparametric}, who showed that output elasticities of capital and labor vary with managerial expense. They also noted that the returns to scale, which is the sum of output elasticities of inputs, are close to one at many managerial expense levels suggesting CRS, but they did not have a test for this hypothesis. Using a newer dataset, our methods can test for CRS while allowing for spatial dependence. We do not reject the CRS hypothesis in all settings, which suggests that CRS technology is a salient feature of this industry. 

Our second study re-examines the Boston house price data of  \cite{harrison1978hedonic}. We estimate a model similar to
\citet{Sun2014}, who used a model selection procedure to recommend that location should have a varying effect on house prices with respect to some other features of the property. We formally test for
the varying location effects of these features using our statistics and indeed find them to be
statistically significant.

In Monte Carlo simulation studies, we  
experiment with a range of SAR models and spatial dependence structures to
demonstrate the applicability and generality of our approach. Results show that our tests have excellent finite-sample performances.

The rest of the paper is as follows: Section \ref{sec:sar} introduces our method in a baseline regression model allowing for higher-order SAR structure. Section \ref{sec:spaterror} extends the model to allow for spatial error dependence of the \citet{Kelejian2007} type. Section \ref{sec:sarlpnp} allows for nonparametric spatial weights, as in \citet{pinkse2002}, and thereby permits some degree of weight matrix robustness. Section \ref{sec:vcsar} shows how our approach can be adapted to the setting where the spatial lag coefficient varies, like in \citet{malikov2017semiparametric}. Section \ref{sec:app} applies our methods to the Chinese nonmetal mineral industry and Boston house prices. Section \ref{sec:mc} presents the simulation study, and Section \ref{sec:con} concludes. Proofs and additional simulation results are in the online appendix.
\section{Baseline higher-order spatial autoregression}\label{sec:sar}
\setlength\abovedisplayskip{6pt}
\setlength\belowdisplayskip{6pt}
Consider a vector of unknown functions $\delta\left(\cdot\right)=\left(\delta_{1}\left(\cdot\right),\ldots,\delta_{d_{\delta}}\left(\cdot\right)\right)'$, and the model 
\begin{equation}\label{sar_model}
	y_{in}=\sum_{j=1}^{d_\lambda}\lambda_j w_{in,j}'y_n+x_{in}'\beta+p_{in}'\delta\left(z_{in}\right)+\epsilon_{in},\irangen,
\end{equation}
where $y_n$ is the $n\times 1$ vector with typical element $y_{in}$, $w_{in,j}'$ is the $i$-th row of the spatial weight matrix $W_{jn}$, $j=1,\ldots,d_\lambda$, $d_\lambda\rightarrow\infty$ as $n\rightarrow\infty$, and $\epsilon_{in}$ is $i.i.d.$ with mean 0 and unit variance. Also,
$\left(x_{in},p_{in},z_{in}\right)\in\calX\times \calP\times\calZ \subseteq\mathbb{R}^{d_{\beta} \times d_{\delta} \times d_z}$, and $d_{\beta}\rightarrow\infty$ as $\ninf$ but $d_z$ and $d_\delta$ are fixed. For all $i=1,\ldots,n$, $z_{in}$ is throughout the paper uncorrelated with $\epsilon_{in}$ while $x_{in}$ is allowed to be correlated with $\epsilon_{in}$. 

It is worth noting that our results also hold when $d_\lambda$ and $d_\beta$ are held fixed but we state our theorems for the more complex case where they diverge. Allowing $d_\lambda$ and $d_\beta$ to diverge with sample size allows a flexible modeling approach, which has been studied in statistics since at least the work of \citet{huber1973robust}, but in spirit the sequence of experiments considered by \citet{LeCam1960} provides an even earlier reference. Econometricians have also studied such models in least squares and GMM settings, see e.g. \citet{Andrews1985}, \citet{Koenker1999} and \citet{Cattaneo2018}, to name a few. In a SAR setting, \citet{Gupta2013} and  \citet{Gupta2018c} emphasize the advantages of this approach, especially when clustered data imply asymptotic regimes where $d_\lambda$ diverges with $n$. 

We now drop $n$ subscripting, but occasionally remind the reader of the dependence on sample size of various quantities. We take spatial weight matrices $W_j$ to be non-stochastic everywhere except in Section \ref{sec:sarlpnp}. These weight matrices are also uniformly bounded in row and column sums, a commonly used restriction to control spatial dependence \citep{kelejian1998generalized, lee2004asymptotic}. This property is assumed to hold almost surely in Section \ref{sec:sarlpnp}, where the spatial weights are stochastic. 

Our aim is to do inference on $\delta(\cdot)$, specifically we wish to test a fixed number $m$ restrictions of the form $S\delta(z)=s, z\in\mathcal{Z}$,
where $S$ is a known, constant $m\times d_\delta$  matrix and $s$ is a known, constant $m\times 1$ vector. Because we can convert any linear restriction to an exclusion restriction, we will focus on tests of the null hypothesis
\begin{equation}\label{truenull_reparameterized}
	H_0^{true}:\delta(z)=0, z\in\mathcal{Z}.
\end{equation} 
Our test statistics will approximate $\delta(\cdot)$ by a series expansion and test if the coefficients in the expansion are jointly zero, thus converting the restriction in $H_0^{true}$ to an increasing-dimensional one, up to some suitably negligible approximation error.

\subsection{Test statistic}

We approximate $\delta_{k}(z)$ by $\psi_{k}^{h_k}(z)'\alpha_{k}^{h_k}$, where $\psi_{k}^{h_{k}}(z)=\left(\psi_{k1}(z),\ldots,\psi_{kh_k}(z)\right)'$ and $\alpha_{k}^{h_{k}}=\left(\alpha_{k1},\ldots,\alpha_{kh_k}\right)'$, for some basis functions $\psi_{k\ell}(\cdot)$, $\ell=1,\ldots,h_k$, $k=1,\ldots,d_{\delta}$. Thus each $\delta_{k}(z)$ is approximated by a linear combination of $h_k$ basis functions $\psi_{k}^{h_{k}}$ with coefficients $\alpha_{k}^{h_{k}}$. Define the $d_\alpha\times 1$ vector $\psi_{i}\equiv\psi_{i}\left(p_{i},z_{i}\right)=\left(p_{i1}\psi_{1}^{h_1}\left(z_{i}\right)',\ldots,p_{i d_{\delta}}\psi_{d_{\delta}}^{h_{d_{\delta}}}\left(z_{i}\right)'\right)'$ and $\alpha=\left(\alpha_{1}^{h_1}{'},\ldots,\alpha_{d_{\delta}}^{h_{d_{\delta}}}{'}\right)'$, with $d_\alpha=\sum_{k=1}^{d_{\delta}}h_k$. Then we can write  (\ref{sar_model}) as
\begin{equation}\label{sarlp_model_approx}
	y_{i}=\sum_{j=1}^{d_\lambda}\lambda_j w_{i,j}'y+x_{i}'\beta+\psi_{i}'\alpha+u_i,\irangen,
\end{equation}
with
$u_i=r_{i}+\epsilon_i$, where   $r_{i}=p_{i}'\delta\left(z_{i}\right)-\psi_{i}'\alpha$ is the approximation error. Let $X$ and $\Psi$ be matrices with typical rows $x_i'$ and $\psi_i'$, and $u$ be with elements $u_i$. Writing $Q$ for the $n\times d_\lambda$ matrix with typical columns $W_jy$, \eqref{sarlp_model_approx} can be written in matrix notation as
\begin{equation}\label{sar_model_approx_matrix}
	y=Q\lambda+X\beta+\Psi\alpha+u=L\xi+u,
\end{equation} 
where $L=[Q, X,\Psi]$ and $\xi=\left(\lambda',\beta',\alpha'\right)'$. The null hypothesis to approximate $H_0^{true}$ is
\begin{equation}\label{null:sar}
	H_0:\alpha=0,
\end{equation}
and our test will be implemented using the 2-stage least squares (2SLS) estimator
\begin{equation}\label{2sls_sar}
	\hat\xi\equiv \left(\hat\lambda',\hat\beta',\hat\alpha'\right)'=\left(L'{\PKa}L\right)^{-1}L'{\PKa}y,	
\end{equation} 
where $\PKa=\Ka\left(\Ka'\Ka\right)^{-1}\Ka'$ and $\Ka$ is an $n\times J_{1}$ instrument matrix
with $J_{1}\geq \dxi=\dlambda+\dbeta+\da$ but the same asymptotic order, i.e. ${J_1}/{\dxi}$ tends to a constant at least unity, as $n\rightarrow \infty$.  Let $\Ta=\invn L'{\PKa}L$ and   
$\cD_{1}=\invn R\Ta^{-1}R'$, where $R=\left[0_{d_\alpha\times\left(d_\lambda+d_\beta\right)}, I_{d_\alpha}\right]$, and $I_{d_\alpha}$ is the $d_\alpha\times d_\alpha$ identity matrix. Then, the test statistic is
\begin{equation}\label{stat:sar}
	\mathbb{W}_1=\frac{n\hat\alpha'\cD_{1}^{-1}\hat\alpha-d_\alpha}{\sqrt{2d_\alpha}}.
\end{equation}

\subsection{Asymptotic properties}
We begin with restrictions on moments of various objects, noting that throughout the paper $C$ denotes a generic positive constant, arbitrarily large but independent of $n$. Let $k_{ri}$ be elements of the instrumental variable matrix, corresponding to $\Ka, \Kb$ and $\Kc$ as defined in the following sections, and $l_{ri}$ denote the elements of $L$.

\begin{assumption}\label{ass:approxerrorsec2}
 $\sup_{i\geq 1}\E\left(r_i^2\right)=o(n^{-1})$.
\end{assumption}

\begin{assumption}\label{ass:errorssec2}
	$\E\left\vert\epsilon_i\right\vert^q<C$, for some $q>4$. 
\end{assumption}

\begin{assumption}\label{ass:secondmomentsec2}
$\E\left(l_{ri}^2\right)<C$ and	$\E\left(k_{ri}^2\right)<C$.
\end{assumption}

\noindent Assumption \ref{ass:approxerrorsec2} controls the approximation error. Sufficient conditions for various cases can be found in \citet{Chen2007}. We avoid using specific rates because many types of approximation errors occur in the paper and explicitly allowing different rates of decay introduces extra notation without adding much insight beyond our general conditions. Assumption \ref{ass:errorssec2} is fairly standard when establishing a central limit theorem for quadratic forms and is used to check a Lyapunov condition. Assumption \ref{ass:secondmomentsec2} imposes moment conditions on regressors and instruments.

To ease notation, we introduce the following definition that captures boundedness and non-multicollinearity of random or fixed matrices of growing dimension.
\begin{definition}\label{def:eigs}
	For any square, symmetric and positive semi-definite matrix $A$, let $\overline\alpha(A)$ and $\underline\alpha(A)$ denote its largest and smallest eigenvalues, respectively. If $A$ is random, we say that $A$ has \textit{Property G} if
	\[
	\overline{\alpha}(A)=O_p(1)\text{ and }\left\{\underline{\alpha}(A)\right\}^{-1}=O_p(1).
	\] 
	If $A$ is non-random, we say that $A$ has \labeltext{\textit{Property G}}\label{propG} if
	\[
	\limsup_{n\rightarrow\infty}\overline{\alpha}(A)<\infty\text{ and }\liminf_{n\rightarrow\infty}\underline{\alpha}(A)>0.
	\] 
\end{definition}

\begin{assumption}\label{ass:eigsec2}
	$n^{-1}\Ka'\Ka$, $n^{-1}L'\Ka\Ka'L$ and $n^{-1}\Psi'\Psi$ have \ref{propG}. 	
\end{assumption}

\noindent We first show that the statistic $\mathbb{W}_1$ can be approximated by a quadratic form in $\epsilon$. Define $\cM_1=\frac{1}{n}\Ka\cV_{1}\Ka'$, with $\cV_{1}=(\Ka'\Ka)^{-1}\Ka'L\Ta^{-1}R'\cD_{1}^{-1}R\Ta^{-1}L'\Ka(\Ka'\Ka)^{-1}$. Then, we have the following theorem.
\begin{theorem}\label{theorem:W1appr}
Under $H_0$, Assumptions \ref{ass:approxerrorsec2}-\ref{ass:eigsec2} and 
    \begin{flalign}\label{rc2.1}
    \frac{1}{\dlambda}+\frac{1}{\dbeta}+\frac{1}{\da}+\frac{\dxi}{n}\rightarrow 0,
    \end{flalign}
    as $n\rightarrow\infty$, we have
	\[
	\mathbb{W}_1-\frac{\epsilon'\mathcal{M}_1\epsilon-d_\alpha}{\sqrt{2d_\alpha}}=o_p(1).
	\]
\end{theorem}
\noindent We remind the reader that results stated in this paper also hold when $d_\lambda$ and $d_\beta$ are held fixed and we do not state separate theorems for that case.
\begin{assumption}\label{ass:Sinv1}
	$\left(I_n-\sum_{j=1}^{d_\lambda} \lambda_j W_j\right)^{-1}$ exists and is uniformly bounded in row and column sums for all sufficiently large $n$.
\end{assumption}

\noindent Define $H_{\ell 1}\equiv H_{\ell 1,n}: \alpha=\alpha^{*}\equiv\nu_{1n}\da^{\frac{1}{4}}/(n\nu_{1n}'\Gamma_{1n}\nu_{1n})^{\frac{1}{2}}$, with $\nu_{1n}$ a $\da\times 1$ non-zero vector and $\Gamma_{1n}$ a $\da\times\da$ matrix defined in detail in the proofs. This sequence of local alternatives features a $d_\alpha^{1/4}$ damping factor that accounts for the cost of our nonparametric approach, and has been found in similar problems by \citet{Hong1995} and \citet{Gupta2018c}, amongst others. We can now state the main theorem of the section.

\begin{theorem}\label{theorem:sar} 
Under Assumptions \ref{ass:approxerrorsec2}-\ref{ass:Sinv1} and 
\begin{flalign}\label{rc2.2}
\frac{1}{\dlambda}+\frac{1}{\dbeta}+\frac{1}{\da}+\frac{\dxi^3}{n}\rightarrow 0,    
\end{flalign}
as $n\rightarrow\infty$, the following hold:
	(i) Under $H_0$, $\mathbb{W}_1\overset{d}{\rightarrow}N(0,1)$.
	(ii) $\mathbb{W}_1$ provides a consistent test. 
	(iii)  Under the sequence of
	local alternatives $H_{\ell 1}$, $\mathbb{W}_1\overset{d}{\rightarrow}N\left(2^{-1/2}, 1\right)$.	
\end{theorem}

\section{SAR with error spatial dependence}\label{sec:spaterror}
Having demonstrated our methods with a baseline model, we now consider (\ref{sar_model}) but relax the $i.i.d$ assumption on $\epsilon_i$ to 
\begin{equation}\label{error_spatial_structure}
	\epsilon_i\equiv \epsilon_{in}=\sum_{j=1}^n b_{ij}v_j,
\end{equation}
where the $v_j$ are $i.i.d$ with mean 0 and variance 1, and $b_{ij}\equiv b_{ijn}$ can depend on $n$. This is the \citet{Kelejian2007} approach to error spatial dependence, and indeed we will employ the spatial heteroskedasticity and autocorrelation consistent (SHAC) method pioneered in that paper. This type of linear process assumption is widely employed in the literature, see e.g. \citet{Robinson2011}, \citet{Robinson2012c}, \citet{Hidalgo2017} and \cite{Conley2023}.
\subsection{Test statistic}

Let $\cD_2=\frac{1}{n}R\cU_1R'$ with $\cU_1= \Ta^{-1}L'\Ka(\Ka'\Ka)^{-1}\Xi_1(\Ka'\Ka)^{-1}\Ka'L\Ta^{-1}$ and $\widehat\cD_2=\frac{1}{n}R\widehat\cU_1R'$ with $\widehat\cU_1=\Ta^{-1}L'\Ka(\Ka'\Ka)^{-1}\hat\Xi_1(\Ka'\Ka)^{-1}\Ka'L\Ta^{-1}$, where $\hat\Xi_1$ is the \citet{Kelejian2007} SHAC estimate of $\Xi_1=\invn\Ka'\Sigma \Ka$, with $\Sigma=E(\epsilon\epsilon')$ being the error variance matrix. We still use the 2SLS estimator $\hat\xi$ and the test statistic is now
\begin{equation}\label{stat:sarlp}
	\mathbb{W}_2=\frac{n\hat\alpha'\widehat\cD_2^{-1}\hat\alpha-d_\alpha}{\sqrt{2d_\alpha}}.
\end{equation}
We introduce the following assumptions:

\begin{assumption}\label{ass:eigsec3}
	$\Sigma$ has \ref{propG}. 	
\end{assumption}

\begin{assumption}\label{ass:bsums}
	$\sup_{i\geq 1}\sum_{j=1}^n\left\vert b_{ij}\right\vert+\sup_{j\geq 1}\sum_{i=1}^n\left\vert b_{ij}\right\vert<\infty$.
\end{assumption}

\begin{assumption}\label{ass:kernel}
	$\fancyk(\cdot): \mathbb{R}\rightarrow [-1,1]$ is a kernel function that satisfies $\fancyk(0)=1$, $\fancyk(x)=\fancyk(-x)$, $\fancyk(x)=0$ for $\vert x \vert>1$ and $\vert \fancyk(x)-1 \vert\leq C\vert x \vert^\varrho, \vert x \vert\leq 1$, for some $\varrho\geq 1$.
\end{assumption}

\noindent Assumption \ref{ass:bsums} restricts the spatial dependence in the errors to a manageable degree, see e.g. \citet{Kelejian2007} and \citet{Delgado2015} for similar assumptions. Assumption \ref{ass:kernel} is a standard assumption on kernels in the HAC setting, see e.g. \citet{Kelejian2007}.

Now we introduce distance measures $d_{ij,m}=d_{ji,m},m=1,\ldots,M$. As in \citet{Kelejian2007}, we allow for
measurement errors and so,  in the following, let $d^*_{ij,m}=d^*_{ji,m}\geq 0$ be the actual distance measures used in practice.
Corresponding to each measure, assume that the researcher can select a distance $d_m>0$ satisfying $d_m\uparrow \infty$ as $n\rightarrow\infty$. For each unit $i=1,\ldots,n$, let $\ell_{i}=\sum_{j=1}^n \left(1-\prod_{m=1}^M\mathbf{1}(d^*_{ij,m}>d_m)\right)$ and set $\ell=\max_{i}\ell_i$. Observe that $\ell_i$ is the number of units $j$ for which $d^*_{ij,m} \leq d_m$ for at least  one $m=1,\ldots,M$. Also let $\sigma_{ij}$ be a typical element of $\Sigma$.
\begin{assumption}\label{ass:distances}
(a) $\E(\ell^2)=o\left(n^{2\eta}\right)$ where $\eta<\frac{1}{2}(q-2) /(q-1)$ with $q>4$ in Assumption \ref{ass:errorssec2}.
	(b) $\sum_{j=1}^n \left\vert \sigma_{ij}\right\vert d^{\chi}_{ij,1}<C$ for some $\chi\geq 1$. (c) $d^*_{ij,m}=d_{ij,m}+\nu_{ij,m}\geq 0$, with $\left\vert \nu_{ij,m}\right\vert<C$ and $\nu_{ij,m}$ independent of $v_i$ for all $m=1,\ldots,M$.	
\end{assumption}

\noindent Then, the $(r,s)$-th element of $\hat\Xi_1$ is
\begin{equation}\label{shac_definition}
	\hat\Xi_{1,rs}=n^{-1}\sum_{i=1}^n\sum_{j=1}^nk_{ri}k_{sj}\hat {u}_{i}\hat {u}_{j}\fancyk\left(\min_m \left\{d^*_{ij,m}/d_m\right\}\right), r,s=1,\ldots,J_1,	
\end{equation}
where $\hat{u}_i$ are elements of the estimation residual vector $\hat u=y-L\hat\xi$.
\subsection{Asymptotic properties}
In all subsequent lemmas/theorems, $q$ is defined in Assumption \ref{ass:errorssec2} and $\eta$ is defined in Assumption \ref{ass:distances}. In Lemma \ref{lemma:shac1} and Theorems \ref{theorem:W2appr}-\ref{theorem:sarlp}, $\dxi$ is defined as $\dxi=\dlambda+\dbeta+\da$. Furthermore, for any matrix $A$, let $\Vert A \Vert=\left\{\overline\alpha\left(A'A\right)\right\}^{\frac{1}{2}}$ i.e. the spectral norm of $A$.
\begin{lemma}\label{lemma:shac1}
Let Assumptions \ref{ass:approxerrorsec2}, \ref{ass:secondmomentsec2}, \ref{ass:eigsec2}, \ref{ass:eigsec3}-\ref{ass:distances} hold, and 
\begin{flalign}\label{rcV1}
 \frac{1}{\dlambda}+\frac{1}{\dbeta}+\frac{1}{\da}+n^{\eta-1}\dxi^2+n^{\eta(q-1)/q-(q-2)/2q}\dxi\rightarrow 0, \text{ as } n\rightarrow\infty.
\end{flalign}
Then, $\norm\hat{\Xi}_1-\Xi_1\norm=o_p(1)$.
\end{lemma}
\noindent Define $\cM_2=\frac{1}{n}B'\Ka\cV_2\Ka'B$, with $\cV_2=(\Ka'\Ka)^{-1}\Ka'L\Ta^{-1}R'\cD_{2}^{-1}R\Ta^{-1}L'\Ka(\Ka'\Ka)^{-1}$. We have the following theorem.
\begin{theorem}\label{theorem:W2appr}
Let (\ref{rcV1}) hold. Then, under $H_0$, Assumptions \ref{ass:approxerrorsec2}, \ref{ass:secondmomentsec2}, \ref{ass:eigsec2}, \ref{ass:eigsec3}-\ref{ass:distances}, as $n\rightarrow\infty$,
	\[
	\mathbb{W}_2-\frac{v'\mathcal{M}_2v-d_\alpha}{\sqrt{2d_\alpha}}=o_p(1).
	\]
\end{theorem}
\noindent To derive the asymptotic distribution of $\mathbb{W}_2$, we need the following assumption, as also in \citet{Delgado2015}. 
\begin{assumption}\label{ass:errorssec3}
	$\E\left\vert v_i\right\vert^s<C$, for some $s\geq 8$. 
\end{assumption}
\noindent Define $H_{\ell 2}\equiv H_{\ell 2,n}: \alpha=\alpha^{*}\equiv\nu_{2n}\da^{\frac{1}{4}}/(n\nu_{2n}'\Gamma_{2n}\nu_{2n})^{\frac{1}{2}}$, with $\nu_{2n}$ a $\da\times 1$ non-zero vector and $\Gamma_{2n}$ a $\da\times\da$ matrix defined in the proof of the next theorem.

\begin{theorem}\label{theorem:sarlp} 
Under Assumptions \ref{ass:approxerrorsec2}, \ref{ass:secondmomentsec2}-\ref{ass:errorssec3}, with
\begin{flalign}\label{rc3.2}
\frac{1}{\dlambda}+\frac{1}{\dbeta}+\frac{1}{\da}+\frac{\dxi^3}{n}+n^{\eta-1}\dxi^2+ \ n^{\eta(q-1)/q-(q-2)/2q}\dxi\rightarrow 0,
\end{flalign}
as $n\rightarrow\infty$, the following hold:
	(i) Under $H_0$,  $\mathbb{W}_2\overset{d}{\rightarrow}N(0,1)$.
	(ii) $\mathbb{W}_2$ provides a consistent test. 
	(iii) Under the sequence of
	local alternatives $H_{\ell 2}$, $\mathbb{W}_2\overset{d}{\rightarrow}N(2^{-1/2}, 1)$.	
\end{theorem}
\noindent Observe that the condition $n^{\eta-1}\dxi^2\rightarrow 0$ implies ${\dxi^3}/{n}\rightarrow 0$ if $\eta\geq 1/3$, so depending on the value of $q$, the rate condition (\ref{rc3.2}) can be simplified.
\section{Misspecification robustness}\label{sec:sarlpnp}
Consider a degree of misspecification robustness in (\ref{sar_model}), where $\epsilon_i$ has the structure in Section \ref{sec:spaterror}. Focusing on the case $d_\lambda=1$ for notational simplicity, and writing $W_1=W$, (\ref{sar_model}) imposed the parametric form $\lambda Wy$ and took $W$ as known in this ``spatial lag'' term. We now allow this to take the nonparametric form $Gy$, where $G$ has elements $g\left(d_{rs}\right)$, $r,s=1,\ldots,n$, for some unknown function $g(\cdot)$, and a vector of exogenous (independent of $v_j$, $j=1,\ldots,n$) economic distance measures $d_{rs}$. This idea was introduced by \citet{pinkse2002}, and has been much used since (see e.g. \citealp{Sun2016} and \citealp{Gupta2024}).

\subsection{Test statistic}
Following \cite{pinkse2002}, we approximate $g\left(d_{ij}\right)$ of $G$ with a series of basis functions
\begin{equation}\label{g_series}
	g(d_{ij})=\sum_{l=1}^{\infty}\tau_le_l(d_{ij}),
\end{equation}
\newcommand{\fC}{\mathfrak{C}}
where $\tau_l$ are unknown coefficients, and $e_{l}$ form a basis of the function space to which $g(\cdot)$ belongs. Let $c_i'$ be the $i$-th row of the $n \times \dtau$ matrix $\fC$ with typical $(i,l)$-th element $\sum_{j\neq i}e_l\left(d_{ij}\right)y_j$ and $\tau=\left(\tau_1,\ldots,\tau_{\dtau}\right)'$, $\dtau\rightarrow\infty$ as $n\rightarrow\infty$. Then (\ref{sar_model})  is extended to
\begin{equation}\label{sarlpnp_model}
	y_{in}=g_{in}'y+x_{in}'\beta+p_{in}'\delta\left(z_{in}\right)+\epsilon_{in},\irangen,
\end{equation} 
and (\ref{sarlp_model_approx}) to
\begin{equation}\label{sarlpnp_model_approx}
	y_{i}=c_i'\tau+x_{i}'\beta+\psi_{i}'\alpha+u_i,\irangen,
\end{equation}
with
$u_i=r_{iG}+r_{i\delta}+\epsilon_i$, where $r_{iG}=\sum_{l=\dtau+1}^{\infty}\tau_l\sum_{j\neq i}e_l\left(d_{ij}\right)y_j$ and  $r_{i\delta}=p_{i}'\delta\left(z_{i}\right)-\psi_{i}'\alpha$. In the matrix notation, we now have
\begin{equation}\label{sarlpnp_model_approx_matrix}
	y=\fC\tau+X\beta+\Psi\alpha+u=F\theta+u,
\end{equation} 
where $F=[\fC, X,\Psi]$ and $\theta=\left(\tau',\beta',\alpha'\right)'$.
Our test is based on the 2SLS estimator
\begin{equation}\label{2sls_sarlpnp}
	\hat\theta=\left(F'{\PKb}F\right)^{-1}F'{\PKb}y,	
\end{equation} 
where $\PKb=\Kb\left(\Kb'\Kb\right)^{-1}\Kb'$, $\Kb$ being an $n\times J_2$ instrument matrix with $J_2\geq \dtheta=\dtau+\dbeta+\da$, but with $J_2$ and $\dtheta$ having the same asymptotic order.  \par

Let $\cD_3=\frac{1}{n}R\cU_2R'$ with $\cU_2=T_{2}^{-1}F'\Kb(\Kb'\Kb)^{-1}\Xi_2(\Kb'\Kb)^{-1}\Kb'FT_{2}^{-1}$ and $T_2=\invn F'{\PKb}F$, and let $\widehat\cD_3=\frac{1}{n}R\widehat\cU_2R'$ with $\widehat\cU_2= T_{2}^{-1}F'\Kb(\Kb'\Kb)^{-1}\hat\Xi_2(\Kb'\Kb)^{-1}\Kb'FT_{2}^{-1}$, where $R=\left[0_{d_\alpha\times\left(d_\tau+d_\beta\right)}, I_{d_\alpha}\right]$ and $\hat\Xi_2$ is the \citet{Kelejian2007} estimate of $\Xi_2=\invn\Kb'\Sigma \Kb$, with $\Sigma=E(\epsilon\epsilon')$. Then, the test statistic is
\begin{equation}\label{stat:sarlpnp}
	\mathbb{W}_3=\frac{n\hat\alpha'\widehat\cD_{3}^{-1}\hat\alpha-d_\alpha}{\sqrt{2d_\alpha}}.
\end{equation}
Denote the elements of $F$ as $f_{ri}$. We have following assumptions.
\begin{assumption}\label{ass:secondmomentsec4}
	$\E\left(f^2_{ri}\right)<C$ and	$\E\left(k_{ri}^2\right)<C$.
\end{assumption}

\begin{assumption}\label{ass:approxerrorsec4}
	$\sup_{i\geq 1}\E\left(r_{iG}^2\right)=o(n^{-1})$ and $\sup_{i\geq 1}\E\left(r_{i\delta}^2\right)=o(n^{-1})$.
\end{assumption}

\begin{assumption}\label{ass:eigsec4}
	$n^{-1}\Kb'\Kb$ and $ n^{-1}F'\Kb\Kb'F$ have \ref{propG}.
\end{assumption}
\begin{lemma}\label{lemma:shac2}
Let Assumptions \ref{ass:eigsec3}-\ref{ass:distances}, \ref{ass:secondmomentsec4}-\ref{ass:eigsec4} hold, and 
\begin{flalign}\label{rcV2}
 \frac{1}{\dtau}+\frac{1}{\dbeta}+\frac{1}{\da}+n^{\eta-1}\dtheta^2+\ n^{\eta(q-1)/q-(q-2)/2q}\dtheta\rightarrow 0, \text{ as } n\rightarrow\infty.
\end{flalign}
Then, $\norm\hat{\Xi}_2-\Xi_2\norm=o_p(1)$.
\end{lemma}

\subsection{Asymptotic properties}
Define $\cM_3=\frac{1}{n}B'\Kb\cV_3\Kb'B$, with $\cV_3=(\Kb'\Kb)^{-1}\Kb'F\Tb^{-1}R'\cD_{3}^{-1}R\Tb^{-1}F'\Kb(\Kb'\Kb)^{-1}$. Then, we have the following theorem.
\begin{theorem}\label{theorem:W3appr}
Let (\ref{rcV2}) hold. Then, under $H_0$, Assumptions \ref{ass:eigsec3}-\ref{ass:distances},  \ref{ass:secondmomentsec4}-\ref{ass:eigsec4}, as $n\rightarrow\infty$, 
	\[
	\mathbb{W}_3-\frac{v'\mathcal{M}_3v-d_\alpha}{\sqrt{2d_\alpha}}=o_p(1).
	\]
\end{theorem}

\noindent Define $H_{\ell 3}\equiv H_{\ell 3,n}: \alpha=\alpha^{*}\equiv\nu_{3n}\da^{\frac{1}{4}}/(n\nu_{3n}'\Gamma_{3n}\nu_{3n})^{\frac{1}{2}}$, with $\nu_{3n}$ a $\da\times 1$ non-zero vector and $\Gamma_{3n}$ a $\da\times\da$ matrix defined in the proofs. The following assumption generalizes Assumption \ref{ass:Sinv1}.
\begin{assumption}\label{ass:Sinv2}
	$\left(I_n-G\right)^{-1}$ exists and is uniformly bounded in row and column sums for all sufficiently large $n$, almost surely and uniformly over the support of the distance measure vector.  
\end{assumption}
\begin{theorem}\label{theorem:sarlpnp} 
Under Assumptions \ref{ass:eigsec3}-\ref{ass:Sinv2}, with
\begin{flalign}\label{rc4.2}
 \frac{1}{\dtau}+\frac{1}{\dbeta}+\frac{1}{\da}+\frac{\dtheta^3}{n}+n^{\eta-1}\dtheta^2+n^{\eta(q-1)/q-(q-2)/2q}\dtheta\rightarrow 0,
\end{flalign}
as $n\rightarrow\infty$, the following hold:
	(1) Under $H_0$, $\mathbb{W}_3\overset{d}{\rightarrow}N(0,1)$.
	(2) $\mathbb{W}_3$ provides a consistent test. 
	(3) Under the sequence of
	local alternatives $H_{\ell 3}$, $\mathbb{W}_3\overset{d}{\rightarrow}N\left(2^{-1/2}, 1\right)$.	
\end{theorem}
\section{Varying spatial coefficients}\label{sec:vcsar}
Now we revert to the case of known $W_j, j=1,\ldots,d_\lambda$, but allow spatial lag parameters $\lambda_j$ to vary. The model is closely related to \citet{malikov2017semiparametric} and \cite{Sun2018}, with the difference that we do not study panel data settings. Instead, we allow $d_\beta\rightarrow\infty$ and general spatial dependence in the errors. Our model also relates to \cite{Sun2024semiparametric}, wherein the model includes a nonparametric nuisance regression function but does not feature dependent errors. Thus, our model complements the extant literature.
\subsection{Test statistic}
Specifically, we have 
\begin{equation}\label{sarvc_model}
	y_{i}=\sum_{j=1}^{d_\lambda}\lambda_j\left(z_i\right) w_{i,j}'y+x_{i}'\beta+p_{i}'\delta\left(z_{i}\right)+\epsilon_{i},\irangen,
\end{equation}
with $\lambda_j(\cdot)$ unknown scalar-valued functions and $\epsilon_i=\sum_{j=1}^n b_{ij}v_j$, i.e. error with spatial dependence. Introduce the approximation $\lambda_j\left(z_i\right)=\sum_{k=1}^{l_j}\mu_{kj}\phi_{kj}\left(z_i\right)+r_{ij\lambda},j=1,\ldots,d_\lambda$. Set $\dl=\sum_{j=1}^{d_\lambda} l_j$, $\phi_j=\left(\phi_{1j},\ldots,\phi_{l_jj}\right)'$, and let $H(z)$ be the $n\times \dl$ matrix with $i$-th row $\left(y'w_{i,1}\phi_1\left(z_i\right),\ldots,y'w_{i,d_\lambda}\phi_{d_\lambda}\left(z_i\right)\right)$. Then,
\begin{equation}\label{sarvc_model_approx_matrix}
	y=G(z)\gamma+u,
\end{equation}
where $G(z)=\left[H(z),X,\Psi\right]$, $\gamma=\left(\mu',\beta',\alpha'\right)'$, $\mu=\left(\mu_1',\ldots,\mu_{\dl}'\right)$, and $u$ has $i$-th element $u_i=r_{i\lambda}+r_{i\delta}+\epsilon_i$ with $r_{i\lambda}=\sum_{j=1}^{d_\lambda}r_{ij\lambda}w_{i,j}'y$. Defining $\Lambda_j(z)=diag\left(\lambda_j\left(z_1\right),\ldots,\lambda_j\left(z_n\right)\right)$ and $S(z)=I-\sum_{j=1}^{d_\lambda}\Lambda_j(z)W_j$, under the usual invertibility conditions, 
\begin{equation}\label{sarvc_RF}
	y=S(z)^{-1}\left(X\beta+t(z)+\epsilon\right),
\end{equation} 
with $P$ having $i$-th row $p_i'$ and $t(z)$ serving to stack $p_i'\delta\left(z_i\right)$ into an $n\times 1$ vector with conformable subscript to $y_i$ and $\epsilon_i$.

The null hypothesis in this setting is
\begin{equation}\label{null_W3}
	H_0:\vartheta=0,
\end{equation}
where $\vartheta$ is either $\mu$ or $\alpha$ and our test will again be implemented using the 2SLS estimator
\begin{equation}\label{2sls_W3}
	\hat\gamma=\left(G'{\PKc}G\right)^{-1}G'{\PKc}y,	
\end{equation} 
where $\PKc=\Kc\left(\Kc'\Kc\right)^{-1}\Kc'$ with $\Kc$ being an $n\times J_3$ instrument matrix with $J_3\geq\dgamma=\dmu+\dbeta+\da$, but with $J_3$ and $\dgamma$ having the same asymptotic order. Note that we have omitted the $z$ argument for brevity. 
Let $\cD_4=\frac{1}{n}R\cU_3R'$ with 
\sloppy
$\cU_3=T_{3}^{-1}G'\Kc(\Kc'\Kc)^{-1}\Xi_3(\Kc'\Kc)^{-1}\Kc'LT_{3}^{-1}$ and $T_3=\invn G'{\PKc}G$, and let 
$\widehat\cD_4=\frac{1}{n}R\widehat\cU_3R'$ with $\widehat\cU_3=T_{3}^{-1}G'\Kc(\Kc'\Kc)^{-1}\hat\Xi_3(\Kc'\Kc)^{-1}\Kc'LT_{3}^{-1}$, where 
$R=\left[I_{\dmu},0_{\dmu\times\left(d_\beta+d_\alpha\right)} \right]$ or $R=\left[0_{\da\times\left(\dmu+\dbeta\right)}, I_{\da}\right]$, 
and $\hat\Xi_3$ is the \citet{Kelejian2007} estimate of $\Xi_3=\invn\Kc'\Sigma \Kc$, with $\Sigma=E(\epsilon\epsilon')$. 
Then, the test statistic is
\begin{equation}\label{W4}
	\mathbb{W}_4=\frac{n\hat\vartheta'\widehat{\cD}_{4}^{-1}\hat\vartheta-\dvartheta}{\sqrt{2\dvartheta}},
\end{equation}
where $\dvartheta$ is either $\dmu$ or $\da$.
\fussy
Denote the elements of $G$ as $g_{ri}(z)$. We have the following assumptions:
\begin{assumption}\label{ass:secondmomentsec5}
	$\sup_{z\in\mathcal Z} \E\left(g^2_{ri}(z)\right)<C$ and	$\E\left(k_{ri}^2\right)<C$.
\end{assumption}

\begin{assumption}\label{ass:approxerrorsec5}
	$\sup_{i\geq 1}\E\left(r_{i\lambda}^2\right)=o(n^{-1})$ and $\sup_{i\geq 1}\E\left(r_{i\delta}^2\right)=o(n^{-1})$.
\end{assumption}

\begin{assumption}\label{ass:eigsec5}
$n^{-1}\Kc'\Kc$ and $ n^{-1}G(z)'\Kc\Kc'G(z)$ have \ref{propG}, the latter uniformly in $z\in \mathcal{Z}$.
\end{assumption}
\begin{lemma}\label{lemma:shac3}
Let Assumptions \ref{ass:eigsec3}-\ref{ass:distances}, and \ref{ass:secondmomentsec5} -\ref{ass:eigsec5} hold, and
\begin{flalign}\label{rcV3}
 \frac{1}{\dmu}+\frac{1}{\dbeta}+\frac{1}{\da}+ n^{\eta-1}\dgamma^2+ n^{\eta(q-1)/q-(q-2)/2q}\dgamma\rightarrow 0, \text{ as } n \rightarrow \infty.
\end{flalign}
Then, we have $\norm\hat{\Xi}_3-\Xi_3\norm=o_p(1)$.
\end{lemma}
\subsection{Asymptotic properties}
Define $\cM_{4}=\frac{1}{n}B'\Kc\cV_{4}\Kc'B$, with $\cV_{4}=(\Kc'\Kc)^{-1}\Kc'G\Tc^{-1}R'\cD_{4}^{-1}R\Tc^{-1}G'\Kc(\Kc'\Kc)^{-1}$. Then, we have the following theorem.  

\begin{theorem}\label{theorem:W4appr}
Let (\ref{rcV3}) hold. Then, under $H_0$, Assumptions \ref{ass:eigsec3}-\ref{ass:distances}, \ref{ass:secondmomentsec5}-\ref{ass:eigsec5}, as $n\rightarrow\infty$,
	\[
	\mathbb{W}_4-\frac{v'\mathcal{M}_4v-\dvartheta}{\sqrt{2\dvartheta}}=o_p(1).
	\]
\end{theorem}
\noindent The following assumption is the appropriate version of Assumptions \ref{ass:Sinv1} and \ref{ass:Sinv2} for the model considered in this section.
\begin{assumption}\label{ass:Sinv3}
	$\left(I_n-S(z)\right)^{-1}$ exists and is uniformly bounded in row and column sums for all sufficiently large $n$, almost surely, uniformly in $z$.   
\end{assumption}
\noindent Define $H_{\ell 4}\equiv H_{\ell 4,n}: \vartheta=\vartheta^{*}\equiv\nu_{4n}\dvartheta^{\frac{1}{4}}/(n\nu_{4n}'\Gamma_{4n}\nu_{4n})^{\frac{1}{2}}$, with $\nu_{4n}$ a $\dvartheta\times 1$ non-zero vector and $\Gamma_{4n}$ a $\dvartheta\times\dvartheta$ matrix defined in the proof of the next theorem.
\begin{theorem}\label{theorem:sarvc} 
\sloppy 
Under Assumptions \ref{ass:eigsec3}-\ref{ass:errorssec3} and \ref{ass:secondmomentsec5}-\ref{ass:Sinv3}, suppose that
\begin{flalign}\label{rc5.2}
 \frac{1}{\dmu}+\frac{1}{\dbeta}+\frac{1}{\da}+\frac{\dgamma^3}{n}+ n^{\eta-1}\dgamma^2+ n^{\eta(q-1)/q-(q-2)/2q}\dgamma\rightarrow 0,
\end{flalign}
as $n\rightarrow\infty$, the following hold:
	(i) Under $H_0$,  $\mathbb{W}_4\overset{d}{\rightarrow}N(0,1)$.
	(ii) $\mathbb{W}_4$ provides a consistent test. 
	(iii) Under the sequence of
	local alternatives $H_{\ell 4}$, $\mathbb{W}_4\overset{d}{\rightarrow}N\left(2^{-1/2},1\right)$.	
\end{theorem}

\allowdisplaybreaks
\section{Empirical applications}
\newcommand{\bK}{\mathbb{K}}
\newcommand{\bL}{\mathbb{L}}
\label{sec:app}

\subsection{Testing for the CRS in a production function}
\label{empcrs} 
In this subsection, we are interested in testing whether the production function in China's nonmetal mineral manufacturing industry has a constant returns to scale (CRS) technology. CRS refers to a scenario where
a simultaneous proportional increase in all inputs results in an identical proportional increase in output. Returns to scale is a classical concept in economics to assess the efficiency of a production function, and it has accordingly received much attention (see %
\citet{ackerberg2015identification,basu2017uncertainty,attanasio2020estimating,combes2021production}%
). Existing works, however, mostly do not account for cross-sectional dependence between firms. The role of cross-firm correlation has recently been highlighted empirically by \cite{iyoha2023estimating}, who shows that this can capture spillover effects and would lead to model misspecification if left unaccounted for.

Our models generalize \cite{li2002semiparametric}'s production function by incorporating spatial dependence. Their analysis centers on a Cobb-Douglas production function that allows elasticities of capital and labor to be nonparametrically varying with management expenses. Management expenses are costs that are indirectly associated with output production, such as research and development
(R\&D), equipment upgrades and employee training. CRS corresponds to the elasticities of capital and labor summing to one.

Using data from the Third Industrial Census of China (conducted by the National Statistical Bureau in 1995), \citet{li2002semiparametric} find that ignoring the varying nature of these coefficients leads to an underestimation of the returns to scale. They suggest their estimate of returns to scale of the technology is close to
being constant for most values of managerial expense, but they do not have a test for the CRS hypothesis. Our data comes from the Chinese Industrial Enterprise Database in 2014 as it is the most recent information available. We use all firms in the nonmetal mineral manufacturing industry with the survey code 31, resulting in a sample of 7,355 observations. As advocated by \cite{li2002semiparametric}, production functions of firms in this industry are expected to be homogeneous since it has a very small proportion of foreign entity
ownership. This is important because evidence suggests that production performance may vary across
ownership types (\citet{murakami1994technical}). 

Variables are in logs: output ($y_{i}$), capital ($p_{i1}$), labor ($p_{i2}$), and management expenses ($z_{i}$). As in \cite{li2002semiparametric}, we include a parametric benchmark model: 
\begin{flalign}\label{CRSbase}
y_{i}=\delta_{0}+\delta_{1}p_{i1}+\delta_{2}p_{i2}+\beta_{1}z_{i}+\beta_{2}z_{i}^2+\epsilon_{i}.
\end{flalign}Elasticities of capital and labor are $\delta
_{1}$ and $\delta _{2}$, respectively. The semiparametric model where the elasticities of of capital
and labor vary with management expenses is: 
\setlength\abovedisplayskip{3pt}
\setlength\belowdisplayskip{3pt}
\begin{flalign*}
	y_{i}=\delta_{0}(z_i)+\delta_{1}(z_i)p_{i1}+\delta_{2}(z_i)p_{i2}+\epsilon_{i}.
\end{flalign*}

\noindent Based on the CRS hypothesis, $H_{0}^{true}:\delta _{1}(z_{i})+\delta _{2}(z_{i})=1$, we reparameterize the model as 
\begin{flalign*}
	y_{i}^{*}=\delta_{0}(z_i)+\delta_{2}(z_i)(p_{i2}-p_{i1})+(\delta_{1}(z_i)+\delta_{2}(z_i)-1)p_{i1}+\epsilon_{i},
\end{flalign*}where $y_{i}^{\ast }=y_{i}-p_{i1}$. Then, we estimate the following model with no SAR structure:
\begin{flalign*}
	\begin{split}			
		& y_{i}^{*}=\beta_0+\suml_{k=1}^{h}\alpha_{1k}\psi_{ik}+\suml_{k=1}^{h}\alpha_{2k}(p_{i2}-p_{i1})\psi_{ik}+\suml_{k=1}^{h}\alpha_{3k}p_{i1}\psi_{ik}+\epsilon_{i}, 
	\end{split}
\end{flalign*} where $\psi _{ik}$ is a basis function, with argument $z_i$, for $k=1,...,h/2$.  We approximate the unknown functions by polynomial and trigonometric functions of different orders. We define the series approximation based Wald statistic derived from this model as $\mathbb{W}_{0}$. 
Next, we estimate a number of models and perform tests based on the various configurations and statistics discussed in previous sections:%
\setlength\abovedisplayskip{5pt}
\setlength\belowdisplayskip{5pt}
\begin{flalign}
\label{eq:sec6.1_sar}
y_{i}^{*}
= \beta_{0}
+ \lambda\suml_{j=1}^{n} w_{ij} y_j
+ \suml_{k=1}^{h}\alpha_{1k}\psi_{ik}
+ \suml_{k=1}^{h}\alpha_{2k}(p_{i2}-p_{i1})\psi_{ik}
+ \suml_{k=1}^{h}\alpha_{3k}p_{i1}\psi_{ik}
+ \epsilon_{i}.
\end{flalign}
We compute the Wald statistics as follows:
$\mathbb{W}_{1}$ from \eqref{eq:sec6.1_sar} with $i.i.d$ errors $\epsilon_{i}$; 
$\mathbb{W}_{2}$ from \eqref{eq:sec6.1_sar} with $\epsilon_{i}=\sum_{j=1}^{n}b_{ij}v_{j};$
$\mathbb{W}_{3}$ from \eqref{eq:sec6.1_sar} replacing $\lambda\sum_{j=1}^{n}w_{ij}y_j$ by
$\sum_{l=1}^{d_\tau}\tau_{l} c_{il}$
and setting $\epsilon_{i}=\sum_{j=1}^{n} b_{ij} v_{j};$
$\mathbb{W}_{4}$ from \eqref{eq:sec6.1_sar} with  
$\lambda\sum_{j=1}^{n}w_{ij}y_j$
replaced by $\sum_{m=1}^{d_\mu}\mu_m \mathfrak{H}_{im}$
and $\varepsilon_i=\sum_{j=1}^{n} b_{ij}v_j$, where $\mathfrak{H}_{im}$ is the $(i,m)$-th entry of the $n\times d_{\mu}$ matrix $H(z)$ defined in Section~\ref{sec:vcsar}.
The coefficient on $p_{i1}$ remains the measure of returns to scale technology even with a SAR feature, although this notion differs from the overall returns to scale that combines technological and other spillovers effects. Let $d_{ij,m}^{\ast }$ be the geographical distance between firms $i$ and $j$, and $d_{m}$ represent the 10th percentile of all geographical distances. We use two different
row-normalized spatial weight matrices, denoted 
\begin{flalign*}
	\begin{split}
		w_{ij}^{p}=[W_{n}^{p}]_{ij}=\begin{cases}
			1  \quad \text{if} \ i \ \text{and} \ j \ \text{ in same city},\\
			0  \quad \text{otherwise}.
		\end{cases} 
		w_{ij}^{d}=[W_{n}^{d}]_{ij}=\begin{cases}
			\frac{1}{d_{ij,m}^{*}}  \quad \text{if} \ d_{ij,m}^{*}<d_m,\\
			0  \quad \text{otherwise}.
		\end{cases} 
	\end{split}
\end{flalign*}

For the test statistics, $\mathbb{W}_{i}$, $i=0,...,3$, $\alpha _{3}=(\alpha _{31},...\alpha
_{3h})^{\prime }$, and the null is $H_{0}:\alpha _{3}=0$. For $\mathbb{W}%
_{4} $, the null is $H_{0}:\alpha _{3}=0$ or $H_{0}:\mu =0$, where $\mu
=(\mu_{1},...,\mu_{d_\mu})^{\prime}$. Throughout our empirical analysis, we use $\psi _{ik}(z)=z_{i}^{k}$ as the
polynomial function for $k=1,...,h$, and $\psi _{ik}(z)=[\sin (kz_{i}),\cos
(kz_{i})]$ as the trigonometric function $k=1,...,h/2$ for $h=2,4$. To
compute $\mathbb{W}_{2}$, $\mathbb{W}_{3}$ and $\mathbb{W}_{4}$, we use the
Epanechnikov kernel function. For $\mathbb{W}_{3}$, let $\mathcal{E}_{l}$ be the matrix whose typical entry is, for $i\neq j$ and $l=1,...,d_\tau$, $e_{l}(d_{ij})=[d_{ij,m}^{\ast }]^{l}\mathbf{1}(d_{ij,m}^{\ast
}<d_{m})$. Then, $c_{il}=\sum_{j=1}^{n}e_{l}(d_{ij})y_{j}$. For $\mathbb{W}_{4}$, we use 
$\phi_{im}(z)=\frac{1}{d_\mu}[\frac{2}{\pi}\tanh(z_{i})]^{m}$ as the polynomial function and $%
\phi_{im}(z)=\frac{1}{d_{\mu}}\sin(\frac{z_{i}}{2m})$ as the trigonometric function for $%
m=1,...,d_\mu$. Then, $\mathfrak{H}_{im}=(\sum_{j=1}^{n}w_{ij}y_{j})\phi_{im}(z)$, where $w_{ij}$ denotes the $(i,j)$-th entry of either $W_n^{p}$ or $W_n^{d}$. 
For the IVs, set $K_{1}=K_{3}=[X,\Psi,W_n^{p}X]$ when the specification uses $W_n^{p}$, and
$K_{1}=K_{3}=[X,\Psi,W_n^{d}X]$ when it uses $W_n^{d}$, for $\mathbb{W}_{1}$, $\mathbb{W}_{2}$, and $\mathbb{W}_{4}$, and $K_{2}=[X, \Psi, \mathcal{E}_{1}X,..., \mathcal{E}_{d_\tau}X]$ with
respect to $\mathbb{W}_{3}$.
We set $d_\mu=d_\tau=h$, and the results with different basis functions are in Table \ref{tab:emp1poly}. 
In all models, the hypothesis of CRS cannot be rejected. These
findings corroborate those in \cite{li2002semiparametric}.
These non-rejection results are economically meaningful. One can conclude that the CRS technology is a salient feature of the production function in this industry, as it is present in the parametric model as well as in semiparametric models with varying coefficients and cross-firm dependence.

We end this application with some remarks: 
(1) In addition to inference on CRS, our analysis also reveals other similarities with prior works. Table \ref{tab:emp1est} shows that the parametric method underestimates the returns to scale compared to the semiparametric method.\footnote{%
The returns to scale in a semiparametric model are computed by summing $\frac{1}{n}\sum_{i=1}^{n}%
\hat{\delta}_{1}(z_{i})$ and $\frac{1}{n}\sum_{i=1}^{n}\hat{\delta}_{2}(z_{i})$.}%
The returns to scale with respect to different input factors also differ, which highlights the crucial role of management expenses.  Figures \ref{fig:crs_comp} (a)-(b) in the online appendix illustrate this with the semiparametric estimates
with spatial weights $W_{n}^{d}$. They show that the output elasticity of capital is increasing in management expenses, while it is decreasing for labor. The decreasing elasticity
of labor is a known empirical fact in this industry due to \textquotedblleft concealed unemployment\textquotedblright\, where the government did not allow state-owned firms to lay off extra employees as a strategy to avoid social unrest.
(2) The returns to scale are relatively flat in the semiparametric model and are within the 95\% bounds for most $%
z_{i}$, as illustrated in Figure \ref{fig:crs_comp} (c) in the online appendix. Further discussions on these findings can be found in \cite{li2002semiparametric}. 
(3) Our results are not unique to the estimates from $W_{n}^{d}$. Other choices of spatial weights give the same conclusion.
\begin{table}[H]
	\centering\footnotesize\stl{3mm}\stla{0mm}
	\caption{Comparison of parametric method and semiparametric method in Section \ref{empcrs}.}
	\begin{spacing}{0.5}
		\begin{tabular}{llcccccc}
			\toprule
			&       & \multicolumn{1}{c}{Without $W$} & \multicolumn{1}{c}{$W_{n}^{p}$} & \multicolumn{1}{c}{$W_{n}^{d}$} & \multicolumn{1}{c}{$\widehat{G}$} & \multicolumn{1}{c}{$W_{n}^{p}$} & \multicolumn{1}{c}{$W_{n}^{d}$} \\
			\midrule
			\multirow{3}[2]{*}{parametric} & $\hat{\delta}_{1}$    & 0.5536  & 0.5438  & 0.5330  & 0.5989  & 0.5996  & 0.5960  \\
			& $\hat{\delta}_{2}$    & 0.4022  & 0.3794  & 0.3725  & 0.3145  & 0.3226  & 0.3099  \\
			& $\hat{\delta}_{1}$+$\hat{\delta}_{2}$ & 0.9558  & 0.9232  & 0.9055  & 0.9134  & 0.9222  & 0.9060  \\
			\midrule
			\multirow{3}[2]{*}{polynomial, $h=2$} & $\hat{\delta}_{1}$   & 0.6228  & 0.6488  & 0.6179  & 0.6263  & 0.6051  & 0.6108  \\
			& $\hat{\delta}_{2}$    & 0.3664  & 0.3212  & 0.3691  & 0.2930  & 0.3502  & 0.3220  \\
			& $\hat{\delta}_{1}+\hat{\delta}_{2}$ & 0.9892  & 0.9699  & 0.9869  & 0.9194  & 0.9553  & 0.9328  \\
			\midrule
			\multirow{3}[2]{*}{trigonometric, $h=2$} & $\hat{\delta}_{1}$    & 0.6233  & 0.6257  & 0.6149  & 0.6284  & 0.6298  & 0.6806  \\
			& $\hat{\delta}_{2}$    & 0.3517  & 0.3099  & 0.3381  & 0.3040  & 0.3377  & 0.2816  \\
			& $\hat{\delta}_{1}+\hat{\delta}_{2}$ & 0.9750  & 0.9356  & 0.9530  & 0.9324  & 0.9675  & 0.9621  \\
			\midrule
			\multirow{3}[2]{*}{polynomial, $h=4$} & $\hat{\delta}_{1}$    & 0.6840  & 0.6390  & 0.5432  & 0.6077  & 0.6213  & 0.6342  \\
			& $\hat{\delta}_{2}$    & 0.3089  & 0.3395  & 0.4345  & 0.3252  & 0.3471  & 0.3025  \\
			& $\hat{\delta}_{1}+\hat{\delta}_{2}$ & 0.9929  & 0.9785  & 0.9777  & 0.9330  & 0.9683  & 0.9368  \\
			\midrule
			\multirow{3}[2]{*}{trigonometric, $h=4$} & $\hat{\delta}_{1}$    & 0.6792  & 0.6896  & 0.6814  & 0.6152  & 0.6289  & 0.6569  \\
			&$\hat{\delta}_{2}$    & 0.3121  & 0.2702  & 0.3002  & 0.3043  & 0.3649  & 0.3199  \\
			& $\hat{\delta}_{1}+\hat{\delta}_{2}$& 0.9913  & 0.9598  & 0.9816  & 0.9195  & 0.9938  & 0.9768  \\
            \bottomrule
		\end{tabular}%
	\end{spacing}
	\label{tab:emp1est}%
\end{table}%

\subsection{Testing the impact of the distance on house prices}
\label{emphouse} 
We use the Boston house price dataset%
\footnote{%
The dataset consists the median
value of owner-occupied homes in 506 census tracts in the Boston Standard
Metropolitan Statistical Area in 1970 (\citet{harrison1978hedonic}), which is available at	\url{https://github.com/simonbrewer/geog6000/blob/main/datafiles/boston.tr.zip}%
.}
to explore how some factors 
affect the median value of owner-occupied homes in $1000$s (MEDV), and whether the effects of these factors
vary over location. Factors include the per
capita crime rate by town (denoted by CRIM), average number of rooms per
dwelling (RM), full-value property-tax rate per \$10,000 dollar
(TAX), percentage of lower socioeconomic status population (LSTAT) and index of accessibility to radial highways (RAD). The location variable DIS, which represents the weighted distances
between a property and the five Boston employment centers, will drive the varying coefficient. 
Let $y_{i},x_{i1},x_{i2},p_{i1},p_{i2},p_{i3},z_{i}$ denote the logarithm of MEDV,
RAD, LSTAT, CRIM, RM, TAX and DIS respectively. 

As house prices in a specific location or region are
commonly influenced by the prices of nearby houses due to factors such as
environmental and geographic considerations, it is natural to consider the spatial autocorrelation.\footnote{%
	There are also some studies that use the SAR semiparametric model with the Boston house price data, for example, \cite%
	{malikov2017semiparametric} and \cite{li2019tests}. However, the former
	considers the covariate $z_{i}$ as NO$_{2}$, which is more about the impact
	of environmental pollution on housing prices. The latter concentrates on the
	geographical weighted regression (GWR) method.} \cite{Sun2014} proposed a
semiparametric coefficient varying spatial model and derived the
distribution theory of their estimator. While they do not have inference results on functions of the random coefficients, they
have a model selection procedure to determine the parametric and
nonparametric components. We thus take the model they
selected as our starting point.

Guided by model selection
results, \cite{Sun2014} argue that DIS doesn't affect LSTAT and RAD. Intuitively,
for LSTAT, the geographical distribution of employment opportunities does
not necessarily correlate with a community's socio-economic status; a
community might be far from employment centers but could have good public
transportation links or a high number of remote work opportunities,
mitigating the impact on the community's lower socioeconomic status. For RAD, the accessibility to highways might depend more on urban planning and infrastructure investment rather than the proximity to employment centers. The baseline model, based on \cite{Sun2014}, is:%
\begin{flalign*}
	y_i=\lambda\suml_{j=1}^{n} w_{ij}y_j+\suml_{k=0}^{2}\beta_{k}x_{ik}+\suml_{m=1}^{3}\delta_{m}(z_{i})p_{im}+\epsilon_{i}.
\end{flalign*} We consider a more general version where the spatial lag coefficient varies with DIS: 
\begin{flalign*}
	y_i=\lambda(z_{i})\suml_{j=1}^{n} w_{ij}y_j+\suml_{k=0}^{2}\beta_{k}x_{ik}+\suml_{m=1}^{3}\delta_{m}(z_{i})p_{im}+\epsilon_{i}.
\end{flalign*}
Let $d_\tau=d_\mu=h$, we estimate: 
\begin{flalign*}
	\begin{split}		
		&y_i=\lambda\suml_{j=1}^{n} w_{ij}y_j+\suml_{k=0}^{2}\beta_{k}x_{ik}+\suml_{m=1}^{3}\suml_{k=1}^{h}\alpha_{mk}p_{im}\psi_{ik}+\epsilon_{i}, \ \text{and using} \ \mathbb{W}_{1};
        \\
		&y_i=\lambda\suml_{j=1}^{n} w_{ij}y_j+\suml_{k=0}^{2}\beta_{k}x_{ik}+\suml_{m=1}^{3}\suml_{k=1}^{h}\alpha_{mk}p_{im}\psi_{ik}+\epsilon_{i}, \  \ \epsilon_{i}=\suml_{j=1}^{n}b_{ij}v_{j} , \ \text{using} \ \mathbb{W}_{2};
        \\
&y_i=\suml_{l=1}^{h}\tau_{l}c_{il}+\suml_{k=0}^{2}\beta_{k}x_{ik}+\suml_{m=1}^{3}\suml_{k=1}^{h}\alpha_{mk}p_{im}\psi_{ik}+\epsilon_{i}, \ \ \epsilon_{i}=\suml_{j=1}^{n}b_{ij}v_{j} , \ \text{using} \ \mathbb{W}_{3};
		\\
		&y_i=\suml_{l=1}^{h}\mu_{l}\mathfrak{H}_{il}+\suml_{k=0}^{2}\beta_{k}x_{ik}+\suml_{m=1}^{3}\suml_{k=1}^{h}\alpha_{mk}p_{im}\psi_{ik}+\epsilon_{i}, \ \ \epsilon_{i}=\suml_{j=1}^{n}b_{ij}v_{j} , \ \text{using} \ \mathbb{W}_{4}.
	\end{split}
\end{flalign*}We consider two row-normalized spatial weight matrices. The first is the first-order queen matrix, $W_{n}^{q}$, discussed in \cite{malikov2017semiparametric}. The second, $W_{n}^{t}$, indicates
whether two locations are in the same tract. For $\mathbb{W}_{i}$, $i=0,...,3$, 
$\alpha=(\alpha_{11},...\alpha_{1h},...,\alpha_{3h})^{\prime }$, and
the null is $H_{0}:\alpha =0$. For $\mathbb{W}_{4}$, the null is $%
H_{0}:\alpha =0$ or $H_{0}:\mu=0$, where 
$\mu=(\mu_{1},...,\mu_{h})^{\prime}$.
Results are listed in Table~\ref{tab:emp2poly}.
All tests reject the null $H_{0}:\alpha =0$. In addition, testing with $%
\mathbb{W}_{4}$ leads to a rejection of $H_{0}:\mu =0$, indicating a significant nonlinear relationship between $log(DIS)$ and neighboring house prices. We further estimate the varying $\lambda (z_{i})$ and evaluate its
empirical mean $\hat{\lambda}=\frac{1}{n}\sum_{i=1}^{n}\lambda (z_{i})$. As shown in Panel E of Table \ref{tab:emp2poly}, the parametric method underestimates the spatial coefficient. Thus, our results based on more general
models complement those in \citet{Sun2014}.

\section{Monte Carlo Simulation}\label{sec:mc}
We present a basic set of Monte Carlo results in this section and a wider range of simulations in the online appendix. All experiments use 1000 replications.

\subsection{Basic setting}\label{basic}
Taking $n = 200, 500, 900,$ we choose two specifications to generate $y$ with $d_{\delta}=1$:
\begin{flalign*}
	y_{i}=\sum_{k=1}^{\dlambda}\lambda_{k}\suml_{j=1}^{n}w_{kij} y_j+x_{i}^{\prime}\beta+p_{i}\delta\left(z_{i}\right)+\epsilon_{i}, i=1, \ldots, n,
\end{flalign*}
and
\begin{flalign*}
	y_{i}=\sum_{k=1}^{\dlambda}\lambda_{k}\suml_{j=1}^{n}w_{kij} y_j+x_{i}^{\prime}\beta+p_{i}\delta\left(z_{i}\right)+v_{i}, \ v_i=\suml_{j=1}^{n}b_{ij}v_j+\epsilon_{i}, i=1, \ldots, n,
\end{flalign*}
where $\epsilon_i$ are $i.i.d$ with three zero mean and unit variance distributions: \labeltext{(V1)}\label{V1} $N(0,1)$, \labeltext{(V2)}\label{V2} $\sqrt{\frac{5}{4}}t(10)$ and \labeltext{(V3)}\label{V3} $\frac{1}{4}(\chi_{8}^{2}-8)$.
We generate $z_i\stackrel{i.i.d}{\sim}U[0,1]$, $p_i\stackrel{i.i.d}{\sim}U[-2,2]$, $x_{i1}=1$, $x_{i2}\stackrel{i.i.d}{\sim}N(1, 2)$ and set $\beta=(-1,1)'$. For $\lambda$, we fix the total weight at 0.9 and then define a decreasing vector $\lambda^{*}=(\dlambda,\dlambda-1,...,1)'$. We then set $\lambda=0.9\left(\sum_{k=1}^{\dlambda}\lambda_{k}^{*}\right)^{-1}\lambda^{*}$.
\par 
For the spatial weights matrices, we generate $W_{k}$ as circulants. Specifically, let $W_{k}^*$ be the symmetric circulant matrix whose first-row entries are
$
w_{k1j}^*= 0$ if  $j=1$ or $j=k+2, \ldots, n-k$ and $w_{k1j}^*= 1$ if $j=2, \ldots, k+1$ or $j=n-k+1, \ldots, n$.
Thus, the weight matrix $W_{k}^*$ encapsulates a binary neighbourhood criterion for $k$ neighbours on either `side' of a unit. 
Now define the normalized matrix $W_{k}=\{\overline{\alpha}(W_{k}^*)\}^{-1} W_{k}^*,$ recalling that $\overline{\alpha}(W_{k}^*)=2k$ for a circulant matrix.
Furthermore, $b_{ij}$ is the typical element of $(I_n-\sum_{k=1}^{\dlambda}\lambda_{k}W_{k})^{-1}$. 
\par 
To construct the SHAC statistics $\mathbb{W}_{2}, \mathbb{W}_{3}$ and $\mathbb{W}_{4}$, we proceed as follows. Generate $M=d_{\lambda}$ distinct distance measures by setting $\ell=[n^{\eta}]+1$ with $\eta=3/7$ and taking $q=8$, so that each unit has at least $\ell$ neighbors, where $\eta$ is defined in Assumption~\ref{ass:distances} (a).\footnote{From the symmetric circulant matrix $W_{\dlambda}$, we build a graph and compute its pairwise shortest‐path distances in MATLAB, then apply \texttt{cmdscale} to generate $n\times 2$ coordinate matrix. We next compute the full $n\times n$ Euclidean distance matrix $D$ with the typical elements $d_{ij}$.} 
For each $m=1, ..., M$, we add the
actual distance matrix, $D_{m}^{*}$, used in practice with unobserved symmetric measurement noise. Its typical entries are
$
d_{ij,m}^{*}=d_{ij}+\nu_{ij,m},$ where $\nu_{ij,m}=1/2(\mu_{ij,m}+\mu_{ji,m})$ with $\mu_{ij,m}\stackrel{i.i.d}{\sim} U[0,1]$.
Define for each unit $i$ the $\ell$-th nearest‐neighbor distance by
$
d_{(i,\ell),m}
=$ the $\ell$-th smallest entry of row $i$ of $D^{*}_{m}$, excluding the diagonal. We then set
$
d_{m}=\max_{1\leq i\leq n} d_{(i,\ell),m},
$
so that every unit has at least $\ell$ neighbors within distance $d_{m}$.  Collecting these over $m=1,\dots,M$ yields the $M\times1$ vector $d_{m}$. We set $\delta=0$ for the null hypothesis and $\delta(\rx)=1-\rx^2$ for any $\rx\in R$ for the alternatives. We choose polynomials $\psi_j(\rx)=\rx^j$ for $j=1,\cdots, h$ and the design of trigonometric functions is in the online appendix. The IV matrix for $\mathbb{W}_{1}$ and $\mathbb{W}_{2}$ is $K_{1}=[X, W_{1}X,..., W_{\dlambda}X,\Psi]$. Furthermore, asy-p represents the standard-normal p-values, whereas chi-p represents the chi-square calibration $(\chi^2_{\da}-\da)/\sqrt{2\da}$, which is to improve testing performance in small samples. This may also help in empirical applications, although in both applications in the previous section critical values from either distribution yield the same results. We discuss size performance in this section and present power performance in the online appendix, since under global alternatives the power of the asymptotic tests tends to one as $n\to\infty$.
\par 
Table~\ref{tab:mcW1size} reports the empirical sizes of $\mathbb{W}_{1}$, and Table~\ref{tab:mcW2size} reports those of the SHAC-corrected statistic $\mathbb{W}_{2}$, at nominal 1\%, 5\% and 10\% levels. When $n=200$, $\mathbb{W}_{2}$ exhibits substantially larger size distortion than $\mathbb{W}_{1}$. The reason is: Theorem~\ref{theorem:W1appr} requires $n^{-1}\,\dxi\to 0,$
whereas Theorem~\ref{theorem:W2appr} requires $n^{\eta-1}\,\dxi\to 0$ with $\eta=3/7$.
In our design $d_{\xi}=d_{\lambda}+d_{\beta}+d_{h}$ with $d_{\lambda}\in\{2,4\}$, $d_{\beta}=2$, $d_{h}\in\{2,4,8\}$, so $d_{\xi}/n$ is already small at $n=200$, but $n^{\eta-1}d_{\xi}^2$ is relatively large. As $n$ increases to 500 and 900, $n^{\eta-1}d_{\xi}^2$ shrinks and the size of $\mathbb{W}_{2}$ converges rapidly to its nominal levels.
Table~\ref{tab:mcW1size}-\ref{tab:mcW2size} also show a clear pattern across the polynomial order with $h\in\{2,4,8\}$.  For both $\mathbb W_{1}$ and $\mathbb W_{2}$, the smallest sieve yields the most accurate sizes, the largest sieve the greatest distortion, particularly for $\mathbb W_{2}$ at $n=200$.  This accords with the fact that larger $h$ increases $d_{\alpha}$ and hence $d_{\xi}=d_{\lambda}+d_{\beta}+d_{\alpha}$, which in turn slows the convergence rates $n^{-1}d_{\xi}\to0$ for $\mathbb W_{1}$ and $n^{\eta-1}d_{\xi}^2\to0$ for $\mathbb W_{2}\,$. A similar scenario obtains when increasing $d_{\lambda}$, since it also raises $d_{\xi}$ and worsens finite-sample size accuracy. 

\subsection{Nonparametric spatial matrix}
For the nonparametric spatial matrix, we generate $g^{*}_{ij}=\Phi\left(-b_{i j}\right) \mathbf{1}\left(c_{ij}<0.1\right)$ if $i \neq j$, and $w_{1ii}=0$, where $\Phi(\cdot)$ is the standard normal cdf, $b_{ij}\stackrel{i.i.d}{\sim} U[-3,3]$, and $c_{i j} \stackrel{i.i.d}{\sim} U[0,1]$. From this construction, we ensure that $G^*$ is sparse with no more than $10\%$ elements being nonzero. Then, define $G=G^{*} / 1.2 \bar{\alpha}(G^{*})$, ensuring the existence of $(I-G)^{-1}$. For $\mathbb{W}_{3}$, we approximate elements in $G$ by $\widehat{g}_{ij}=\sum_{l=0}^{\dtau}\tau_{l}e_{l}(d_{ij})=\sum_{l=0}^{\dtau}\tau_{l}[d_{ij,m}^{*}]^{l}\mathbf{1}(d_{ij,m}^{*}<d_{m})$ if $i \neq j$ with $\dtau=2,4$, where $d_{ij,m}^{*}$ and $d_{m}$ are discussed in Section \ref{basic}, and the IV matrix for $\mathbb{W}_{3}$ is $K_{2}=[X,\mathcal{E}_{1}X,...,\mathcal{E}_{\dtau}X,\Psi]$, where $e_{l}(d_{ij})$ are the typical elements of $\mathcal{E}_{l}$.\par 
Table~\ref{tab:mcW3size} reports the empirical sizes of $\mathbb{W}_{3}$ at nominal 1\%, 5\% and 10\% levels.  When $n=200$, $\mathbb{W}_{3}$ with $\dtau=2$ yields rejection rates very close to the nominal levels, whereas the larger basis $\dtau=8$ produces mild over-rejection.  As $n$ increases to 500 and 900, both choices of $\dtau$ rapidly converge to their nominal sizes. This scenario is similar to that of $\mathbb{W}_{2}$ and originates from the rate condition $n^{\eta-1}d_{\tau}^2\to0$ in Theorem \ref{theorem:W3appr}. In other words, it also requires an accurate approximation of the SHAC estimator. 

\subsection{Varying spatial coefficient}
We now consider varying $\lambda$ such that $\lambda=\lambda(z_i)$ and $d_{\lambda}=1,2$. When $d_{\lambda}=1$, we set $\lambda(z_i)=0.9\sin(\pi z_i)$, when $d_{\lambda}=2$, we set $\lambda_1(z_i)=0.6\sin(\pi z_i)$ and $\lambda_2(z_i)=0.3\sin(\pi z_i)$. Others are the same as those in Section \ref{basic}. For estimation, 
we use $\phi_{il}(z)=\frac{1}{h}[\frac{2}{\pi}\tanh(z_{i})]^{l}$ as the polynomial function and $\phi_{il}(z) =\frac{1}{h}\sin(\frac{z_i}{2l})$ as the trigonometric function for $l=1,...,h$.
The IV matrix $K_{3}=K_{1}$ with $d_{\lambda}=1,2$.\par 
Table~\ref{tab:mcW4size} reports the rejection probabilities of $\mathbb{W}_{4}$ at nominal 1\%, 5\% and 10\% levels when testing $\delta(z)$.  At $n=200$, the smallest sieve $(h=2,\,d_{\lambda}=1)$ achieves sizes close to nominal, whereas larger sieves, especially $h=8,\,d_{\lambda}=2$, over‐reject noticeably. As $n$ increases to 900, size distortion vanishes rapidly.  Table~\ref{tab:mcW4size_lambda} shows the sizes of $\mathbb W_{4}$ when testing $\lambda(z)$.  In particular, at $n=200$ with $d_{\lambda}=2$ and $h=8$, the over‐rejection is especially severe. By $n=900$, size accuracy is largely restored. This pattern reflects the rate condition in Theorem~\ref{theorem:W4appr},
$n^{\eta-1}d_{\gamma}^2 \to 0$,
since both $h$ and $d_{\lambda}$ enter the total sieve dimension $d_{\gamma}$, slowing the approximation when they are large. Across all designs, the chi-p delivers empirical sizes closer to nominal levels than asy-p in small samples, while the difference dissipates as $n$ grows. 
Results for trigonometric functions are similar, and left for the online appendix.

\section{Conclusion}\label{sec:con}

We provide a machinery for conducting tests on nonparametrically varying coefficients using ideas from parametric testing, and with similar ease of implementation. We allow for spatial dependence and provide for testing nonparametrically varying spatial dependence coefficients. It seems reasonable to conjecture that the methodology is applicable in other settings where one may be interested in testing linear restrictions on nonparametric functions. These provide directions for future research.

\setlength{\bibsep}{0.3ex}
\begin{spacing}{1}\small
	\normalem
\bibliographystyle{chicago}
\bibliography{AG_refs}
\end{spacing}

\begin{table}[H]
\centering\footnotesize\stl{1.5mm}\stla{0mm}
\caption{Estimation and testing results in Section \ref{empcrs}.}
	\begin{spacing}{0.4}
\begin{tabular}{lcccccccc}
			\toprule
        			\multicolumn{9}{c}{Panel A: polynomial series, $h=2$}\\
			\midrule 
			Matrices &	& \multicolumn{1}{c}{$W_{n}^{p}$} & \multicolumn{1}{c}{$W_{n}^{d}$} & \multicolumn{1}{c}{$W_{n}^{p}$} & \multicolumn{1}{c}{$W_{n}^{d}$} & \multicolumn{1}{c}{$\widehat{G}$}       & \multicolumn{1}{c}{$W_{n}^{p}$} & \multicolumn{1}{c}{$W_{n}^{d}$} \\
			\midrule
			$\lambda/\tau_1/\mu_1$    &       & \multicolumn{1}{l}{0.189*} & \multicolumn{1}{l}{0.130$^{**}$} & \multicolumn{1}{l}{0.224$^{***}$} & \multicolumn{1}{l}{0.256$^{***}$} & \multicolumn{1}{l}{0.451$^{***}$} & \multicolumn{1}{l}{0.457$^{***}$} & \multicolumn{1}{l}{0.499$^{***}$} \\
			&       & \multicolumn{1}{l}{(0.112)} & \multicolumn{1}{l}{(0.062)} & \multicolumn{1}{l}{(0.014)} & \multicolumn{1}{l}{(0.081)} & \multicolumn{1}{l}{(0.157)} & \multicolumn{1}{l}{(0.042)} & \multicolumn{1}{l}{(0.028)} \\
			$\tau_2/\mu_2$    &       &       &       &       &       & \multicolumn{1}{l}{-0.095} & \multicolumn{1}{l}{0.405$^{***}$} & \multicolumn{1}{l}{0.484$^{***}$} \\
			&       &       &       &       &       & \multicolumn{1}{l}{(0.126)} & \multicolumn{1}{l}{(0.047)} & \multicolumn{1}{l}{(0.030)} \\
			\midrule
			& \multicolumn{1}{c}{$\mathbb{W}_{0}$} & \multicolumn{1}{c}{$\mathbb{W}_{1}$} & \multicolumn{1}{c}{$\mathbb{W}_{1}$} & \multicolumn{1}{c}{$\mathbb{W}_{2}$} & \multicolumn{1}{c}{$\mathbb{W}_{2}$} & \multicolumn{1}{c}{$\mathbb{W}_{3}$} & \multicolumn{1}{c}{$\mathbb{W}_{4}$} & \multicolumn{1}{c}{$\mathbb{W}_{4}$}\\
			\midrule
			$H_0:\alpha_3=0$    & 1.593  & 3.311  & 3.280  & 1.419  & 4.020  & 0.226  & 1.716  & 2.589  \\
\end{tabular}%
	\end{spacing}
	\label{tab:emp1poly}%
\end{table}%

\begin{table}[H]
\centering\footnotesize\stl{1.5mm}\stla{0mm}
\caption{Estimation and testing results in Section \ref{emphouse}.}
	\begin{spacing}{0.5}
    		\begin{tabular}{llcccccc}
			\toprule
        			\multicolumn{8}{c}{Panel A: polynomial series $h=2$}\\
			\midrule 
			Matrices  &\multicolumn{1}{c}{$W_{n}^{q}$} & \multicolumn{1}{c}{$W_{n}^{t}$} & \multicolumn{1}{c}{$W_{n}^{q}$} & \multicolumn{1}{c}{$W_{n}^{t}$} & \multicolumn{1}{c}{$\widehat{G}$}       & \multicolumn{1}{c}{$W_{n}^{q}$} & \multicolumn{1}{c}{$W_{n}^{t}$} \\
			\midrule
			$\lambda/\tau_1/\mu_{1}$    &      \multicolumn{1}{l}{0.238$^{***}$} & \multicolumn{1}{l}{0.219$^{***}$} & \multicolumn{1}{l}{0.238$^{***}$} & \multicolumn{1}{l}{0.219$^{***}$} & \multicolumn{1}{l}{0.214$^{**}$} & \multicolumn{1}{l}{0.512$^{***}$} & \multicolumn{1}{l}{0.468$^{***}$} \\
			&    \multicolumn{1}{l}{(0.049)} & \multicolumn{1}{l}{(0.044)} & \multicolumn{1}{l}{(0.028)} & \multicolumn{1}{l}{(0.022)} & \multicolumn{1}{l}{(0.084)} & \multicolumn{1}{l}{(0.073)} & \multicolumn{1}{l}{(0.089)} \\
			$\tau_2/\mu_{2}$    &          &       &       &       & \multicolumn{1}{l}{0.330$^{**}$} & \multicolumn{1}{l}{-0.356$^{***}$} & \multicolumn{1}{l}{-0.311$^{***}$} \\
			&       &           &       &       & \multicolumn{1}{l}{(0.129)} & \multicolumn{1}{l}{(0.049)} & \multicolumn{1}{l}{(0.057)} \\
			\midrule
			Testing statistics&  $\mathbb{W}_{1}$ & $\mathbb{W}_{1}$ & $\mathbb{W}_{2}$ & $\mathbb{W}_{2}$ & $\mathbb{W}_{3}$ & $\mathbb{W}_{4}$ & $\mathbb{W}_{4}$\\
			\midrule
		  $H_0:\alpha=0$  & \multicolumn{1}{l}{40.041$^{***}$} & \multicolumn{1}{l}{33.913$^{***}$} & \multicolumn{1}{l}{290.165$^{***}$} & \multicolumn{1}{l}{33.913$^{***}$} & \multicolumn{1}{l}{183.122$^{***}$} & \multicolumn{1}{l}{91.459$^{***}$} & \multicolumn{1}{l}{123.439$^{***}$} \\
\end{tabular}%
	\end{spacing}
	\label{tab:emp2poly}%
\end{table}%

\begin{table}[H]
	\centering\footnotesize\stl{3mm}\stla{0mm}
	\caption{Empirical sizes for $\mathbb{W}_{1}$, rejection probabilities at 1\%, 5\% and 10\% level.}
	\begin{spacing}{1}
    \begin{tabular}{lllrrrrrrrrr}
    \toprule\toprule
    \multicolumn{3}{c}{\multirow{2}[3]{*}{$\dlambda=2$}} & \multicolumn{3}{c}{$n=200$} & \multicolumn{3}{c}{$n=500$} & \multicolumn{3}{c}{$n=900$} \\
\cmidrule{4-12}    \multicolumn{3}{c}{}  & \multicolumn{1}{c}{1\%} & \multicolumn{1}{c}{5\%} & \multicolumn{1}{c}{10\%} & \multicolumn{1}{c}{1\%} & \multicolumn{1}{c}{5\%} & \multicolumn{1}{c}{10\%} & \multicolumn{1}{c}{1\%} & \multicolumn{1}{c}{5\%} & \multicolumn{1}{c}{10\%} \\
    \midrule\midrule
     \multirow{6}[2]{*}{\ref{V1}} & \multirow{2}[1]{*}{$h=2$} & asy-p & 0.035  & 0.075  & 0.115  & 0.033  & 0.063  & 0.098  & 0.042  & 0.074  & 0.107  \\
          &       & chi-p & 0.013  & 0.051  & 0.112  & 0.006  & 0.049  & 0.096  & 0.010  & 0.056  & 0.106  \\
          & \multirow{2}[0]{*}{$h=4$} & asy-p & 0.033  & 0.074  & 0.112  & 0.031  & 0.074  & 0.107  & 0.038  & 0.077  & 0.114  \\
          &       & chi-p & 0.013  & 0.056  & 0.106  & 0.017  & 0.050  & 0.105  & 0.010  & 0.055  & 0.106  \\
          & \multirow{2}[1]{*}{$h=8$} & asy-p & 0.033  & 0.073  & 0.117  & 0.025  & 0.060  & 0.091  & 0.011  & 0.060  & 0.096  \\
          &       & chi-p & 0.014  & 0.057  & 0.108  & 0.011  & 0.046  & 0.088  & 0.008  & 0.040  & 0.091  \\
\end{tabular}%
	\end{spacing}
	\label{tab:mcW1size}%
\end{table}%

\begin{table}[H]
	\centering\footnotesize\stl{3mm}\stla{0mm}
	\caption{Empirical sizes for $\mathbb{W}_{2}$, rejection probabilities at 1\%, 5\% and 10\% level.}
	\begin{spacing}{1}
    \begin{tabular}{lllrrrrrrrrr}
    \toprule\toprule
    \multicolumn{3}{c}{\multirow{2}[3]{*}{$\dlambda=2$}} & \multicolumn{3}{c}{$n=200$} & \multicolumn{3}{c}{$n=500$} & \multicolumn{3}{c}{$n=900$} \\
\cmidrule{4-12}    \multicolumn{3}{c}{}  & \multicolumn{1}{c}{1\%} & \multicolumn{1}{c}{5\%} & \multicolumn{1}{c}{10\%} & \multicolumn{1}{c}{1\%} & \multicolumn{1}{c}{5\%} & \multicolumn{1}{c}{10\%} & \multicolumn{1}{c}{1\%} & \multicolumn{1}{c}{5\%} & \multicolumn{1}{c}{10\%} \\
    \midrule\midrule
    \multirow{6}[2]{*}{\ref{V1}} & \multirow{2}[1]{*}{$h=2$} & asy-p & 0.037  & 0.071  & 0.099  & 0.033  & 0.061  & 0.095  & 0.048  & 0.085  & 0.124  \\
          &       & chi-p & 0.008  & 0.055  & 0.096  & 0.007  & 0.040  & 0.092  & 0.016  & 0.061  & 0.124  \\
          & \multirow{2}[0]{*}{$h=4$} & asy-p & 0.046  & 0.083  & 0.109  & 0.030  & 0.064  & 0.097  & 0.046  & 0.078  & 0.118  \\
          &       & chi-p & 0.024  & 0.057  & 0.108  & 0.010  & 0.046  & 0.093  & 0.016  & 0.061  & 0.115  \\
          & \multirow{2}[1]{*}{$h=8$} & asy-p & 0.079  & 0.111  & 0.138  & 0.031  & 0.059  & 0.079  & 0.026  & 0.047  & 0.082  \\
          &       & chi-p & 0.053  & 0.105  & 0.132  & 0.012  & 0.046  & 0.078  & 0.011  & 0.036  & 0.085  \\
\end{tabular}%
	\end{spacing}
	\label{tab:mcW2size}%
\end{table}%

\begin{table}[H]
	\centering\footnotesize\stl{3mm}\stla{0mm}
	\caption{Empirical sizes for $\mathbb{W}_{3}$, rejection probabilities at 1\%, 5\% and 10\% level.}
	\begin{spacing}{1}
    \begin{tabular}{lllrrrrrrrrr}
    \toprule\toprule
    \multicolumn{3}{c}{\multirow{2}[3]{*}{$\dtau=2$}} & \multicolumn{3}{c}{$n=200$} & \multicolumn{3}{c}{$n=500$} & \multicolumn{3}{c}{$n=900$} \\
\cmidrule{4-12}    \multicolumn{3}{c}{}  & \multicolumn{1}{c}{1\%} & \multicolumn{1}{c}{5\%} & \multicolumn{1}{c}{10\%} & \multicolumn{1}{c}{1\%} & \multicolumn{1}{c}{5\%} & \multicolumn{1}{c}{10\%} & \multicolumn{1}{c}{1\%} & \multicolumn{1}{c}{5\%} & \multicolumn{1}{c}{10\%} \\
    \midrule\midrule
    \multirow{6}[2]{*}{\ref{V1}} & \multirow{2}[1]{*}{$h=2$} & asy-p & 0.037  & 0.071  & 0.099  & 0.033  & 0.061  & 0.095  & 0.048  & 0.085  & 0.124  \\
          &       & chi-p & 0.008  & 0.055  & 0.096  & 0.007  & 0.040  & 0.092  & 0.016  & 0.061  & 0.124  \\
          & \multirow{2}[0]{*}{$h=4$} & asy-p & 0.046  & 0.083  & 0.109  & 0.030  & 0.064  & 0.097  & 0.046  & 0.078  & 0.118  \\
          &       & chi-p & 0.024  & 0.057  & 0.108  & 0.010  & 0.046  & 0.093  & 0.016  & 0.061  & 0.115  \\
          & \multirow{2}[1]{*}{$h=8$} & asy-p & 0.079  & 0.111  & 0.138  & 0.031  & 0.059  & 0.079  & 0.026  & 0.047  & 0.082  \\
          &       & chi-p & 0.053  & 0.105  & 0.132  & 0.012  & 0.046  & 0.078  & 0.011  & 0.036  & 0.085  \\
\end{tabular}%
	\end{spacing}
	\label{tab:mcW3size}%
\end{table}%

\begin{table}[H]
	\centering\footnotesize\stl{3mm}\stla{0mm}
	\caption{Empirical sizes for $\mathbb{W}_{4}$, rejection probabilities at 1\%, 5\% and 10\% level. Testing $\delta(z)$.}
	\begin{spacing}{1}
    \begin{tabular}{lllrrrrrrrrr}
    \toprule\toprule
    \multicolumn{3}{c}{\multirow{2}[3]{*}{$\dlambda=1$}} & \multicolumn{3}{c}{$n=200$} & \multicolumn{3}{c}{$n=500$} & \multicolumn{3}{c}{$n=900$} \\
\cmidrule{4-12}    \multicolumn{3}{c}{}  & \multicolumn{1}{c}{1\%} & \multicolumn{1}{c}{5\%} & \multicolumn{1}{c}{10\%} & \multicolumn{1}{c}{1\%} & \multicolumn{1}{c}{5\%} & \multicolumn{1}{c}{10\%} & \multicolumn{1}{c}{1\%} & \multicolumn{1}{c}{5\%} & \multicolumn{1}{c}{10\%} \\
    \midrule\midrule
    \multirow{6}[2]{*}{\ref{V1}} & \multirow{2}[1]{*}{$h=2$} & asy-p & 0.065  & 0.080  & 0.129  & 0.053  & 0.086  & 0.110  & 0.050  & 0.072  & 0.090  \\
          &       & chi-p & 0.022  & 0.053  & 0.124  & 0.027  & 0.069  & 0.109  & 0.021  & 0.059  & 0.089  \\
          & \multirow{2}[0]{*}{$h=4$} & asy-p & 0.042  & 0.085  & 0.112  & 0.071  & 0.106  & 0.130  & 0.047  & 0.072  & 0.098  \\
          &       & chi-p & 0.039  & 0.071  & 0.103  & 0.026  & 0.090  & 0.126  & 0.029  & 0.061  & 0.092  \\
          & \multirow{2}[1]{*}{$h=8$} & asy-p & 0.169  & 0.124  & 0.159  & 0.019  & 0.060  & 0.074  & 0.022  & 0.039  & 0.071  \\
          &       & chi-p & 0.144  & 0.101  & 0.154  & 0.025  & 0.054  & 0.071  & 0.016  & 0.036  & 0.080  \\
\end{tabular}%
	\end{spacing}
	\label{tab:mcW4size}%
\end{table}%

\newpage
\appendix
\setcounter{assumption}{0}
\renewcommand{\theassumption}{\Alph{assumption}}%
\setcounter{equation}{0}
\renewcommand{\theequation}{\thesection.\arabic{equation}}%
\setstretch{1.2}
\begin{center}
    {\Large \textbf{A Supplement to}\\ \titlename}
\end{center}
\begin{center}
    \begin{minipage}{0.9\textwidth}
\centering
\textbf{Abhimanyu Gupta}\\{\small Department of Economics, Queen's University, Dunning Hall, 94 University Avenue, Kingston, Ontario K7L 3N6, Canada and Department of Economics, University of Essex, Wivenhoe Park, Colchester, CO4 3SQ, UK. Email: abhimanyu.g@queensu.ca.} 
\\
\textbf{Xi Qu}\\
{\small Department of Economics, Antai College of Economics and Management, Shanghai Jiao Tong University, 1954 Huashan Road, Shanghai, 200030, China PRC. Email: xiqu@sjtu.edu.cn.}
\\
\textbf{Sorawoot Srisuma}\\
{\small Department of Economics, National University of Singapore, 1 Arts Link, 117570, Singapore.} 
\\
\textbf{Jiajun Zhang}\\
{\small International Business School, Shanghai University of International Business and Economics, 201620, China PRC. E-mail: jiajun30@suibe.edu.cn.}
\end{minipage}
\end{center}
\setstretch{1.2}

\section{General Central Limit Theorem}\label{sec:app_genCLT}
Introduce an $n\times n$ matrix $\fancya$, which is symmetric, idempotent and has rank $J\rightarrow \infty$. 
\begin{assumption}\label{ass:clt_general}
Assume that 

(a)
\begin{equation}\label{fancya_definition}
\fancya=B'\fancys B,	
\end{equation}
where $B$ is the $n\times n$ matrix with entries $b_{ij}$ satisfying Assumption \ref{ass:bsums} and $\fancys$ is a symmetric matrix. (b) The entries $s_{ij}$ of $\fancys$ satisfy $s_{ij}=O_p(J/n)$ and $\sum_{i=1}^n s_{ij}^2=O_p(J/n)$, uniformly in $i,j=1,\ldots,n$. (c) $\fancys$ satisfies \ref{propG}. 
\end{assumption}
\begin{theorem}\label{thm:general_clt}
Let $v$ be an $n \times 1$ $i.i.d. (0,1)$ vector with elements satisfying Assumption \ref{ass:errorssec3}. Under Assumption \ref{ass:clt_general} and $J^{-1}+J^3/n\rightarrow 0$ as $n\rightarrow\infty$,
	\begin{equation}
		{\left(v ^{\prime }\fancya v -J\right)}/{\sqrt{2J}}\overset{d}{\longrightarrow}N(0,1).\label{clt_tgt2}
	\end{equation}
\end{theorem}

\begin{proof}
${\left(v ^{\prime }\fancya v-J\right)}/{\sqrt{2J}}=\sum_{s=1}^{n }\rho_{s}$, where
\begin{center}
$\rho_{s}=L^{-1}a_{ss}\left(  v _{s}^{2}-1\right) +2L^{-1}\mathbf{1}(s\geq 2) v
_{s}\sum_{t<s}a_{st} v _{t},$
\end{center} writing $L=\sqrt{2J}$.
From \cite{Scott1973}, (\ref{clt_tgt2}) follows if
\begin{equation}
		\sum_{s=1}^{n}\E \rho_{s}^{4}\overset{p}{\longrightarrow }0,\text{ as }%
		n\rightarrow \infty ,  \label{pmle_lyap}
\end{equation}%
\begin{equation}
\text{and }\sum_{s=1}^{n}\left[ \E \left( \rho_{s}^{2}\left. {}\right \vert v _{t},t<s\right) -\E \left( \rho_{s}^{2}\right) \right] \overset{%
			p}{\longrightarrow }0,\text{ as }n\rightarrow \infty .  \label{pmle_var}
\end{equation}%
We show (\ref{pmle_lyap}) first. Evaluating the LHS and using Assumption \ref{ass:clt_general} yields
\begin{eqnarray*}
		\E \rho_{s}^{4} &\leq &CL^{-4}a_{ss}^{4}+CL^{-4}\sum_{t<s}a_{st}^{4}\leq
		CL^{-4}\left( \sum_{t\leq s}a_{st}^{2}\right) ^{2} \leq CL^{-4}\left( b_{s}^{\prime }\fancys\sum_{t\leq s}b_{t}b_{t}^{\prime
		}\fancys b_{s}\right) ^{2}\\
		&\leq& CL^{-4}\left( b_{s}^{\prime }\fancys^{2}b_{s}\right)
		^{2}= CL^{-4}\sum_{i,j,k=1}^n b_{is}b_{ks}s_{ij}s_{kj} \\
		&\leq& CL^{-4}\sum_{i,k=1}^n \left \vert b_{is}\right\vert \left\vert b_{ks}\right\vert\sum_{j=1}^n \left(s^2_{ij}+s^2_{kj}\right)=O_{p}\left( L^{-4}Jn^{-1}\left(\sum_{i=1}^n \left \vert b_{is}\right\vert\right)^2\right) ,
\end{eqnarray*}%
so $\sum_{s=1}^{n}\E \rho_{s}^{4} =O_{p}\left(
		L^{-4}Jn^{-1}\sum_{s=1}^{n}\left(\sum_{i=1}^n \left \vert b_{is}\right\vert\right)^2\right)
		=O_{p}\left( L^{-4}J\right) $.
Thus (\ref{pmle_lyap}) is established. Note that 
$\E \left( \left. \rho_{s}^{2}\right \vert v
	_{t},t<s\right)=
		4L^{-2}\left \{ \left( \mu _{4}-1\right)
		a_{ss}^{2}+2\mu _{3}\mathbf{1}(s\geq 2)\sum_{t<s}a_{st}a_{ss} v
		_{t}\right \} +4L^{-2}\mathbf{1}(s\geq 2)\left(
		\sum_{t<s}a_{st} v _{t}\right) ^{2},$ where $\E v_t^j=\mu_j, j=3, 4$,
	and $
	\E \rho_{s}^{2}=4L^{-2}\left( \mu _{4}-1\right) a_{ss}^{2}+4L^{-2}\mathbf{1}(s\geq 2)\sum_{t<s}a_{st}^{2},
	$
so that (\ref{pmle_var}) is bounded by 
\begin{equation}
		CL^{-2}\sum_{s=2}^{n}\sum_{t<s}a_{st}a_{ss} v _{t}+C\left \{
		\sum_{s=2}^{n}\left( \sum_{t<s}a_{st} v _{t}\right) ^{2}-\sum_{t<s}a_{st}^{2}\right \} .  \label{pmle_var_11_1}
\end{equation}%
By transforming the range of summation, the square of the first term in (\ref%
	{pmle_var_11_1}) has expectation bounded by
	\begin{equation}
		CL^{-4}\E \left( \sum_{t=1}^{n-1}\sum_{s=t+1}^{n}a_{st}a_{ss} v
		_{t}\right) ^{2}\leq CL^{-4}\sum_{t=1}^{n-1}\left(
		\sum_{s=t+1}^{n}a_{st}a_{ss}\right) ^{2},  \label{pmle_var_11_2}
	\end{equation}%
where the factor in parentheses on the RHS of (\ref{pmle_var_11_2}) is
\begin{eqnarray*}
		&&\sum_{s,r=t+1}^{n}b_{s}^{\prime }\fancys b_{s}b_{s}^{\prime }\fancys b_{t}b_{r}^{\prime
		}\fancys b_{r}b_{r}^{\prime }\fancys b_{t} \leq\sum_{s,r=t+1}^{n}\left \vert b_{s}^{\prime }\fancys b_{s}b_{r}^{\prime
		}\fancys b_{r}\right \vert \left \vert b_{s}^{\prime }\fancys b_{t}\right \vert \left
		\vert b_{r}^{\prime }\fancys b_{t}\right \vert \\
		&\leq & C\sum_{s,r=t+1}^{n}\sum_{i,j,k,l=1}^{n}\left \vert b_{is}\right \vert
		\left \vert s_{ij}\right \vert \left \vert b_{jr}\right \vert \left \vert
		b_{ks}\right \vert \left \vert s_{lk}\right \vert \left \vert b_{kr}\right
		\vert \left \vert b_{s}^{\prime }\fancys b_{t}\right \vert \left \vert
		b_{r}^{\prime }\fancys b_{t}\right \vert \\
		&\leq &C\left( \sup_{i,j}\left \vert s_{ij}\right \vert \right) ^{2}\left(
		\sup_{s\geq 1}\sum_{i=1}^{n}\left \vert b_{is}\right \vert \right)
		^{4}\sum_{s,r=t+1}^{n}\left \vert b_{s}^{\prime }\fancys b_{t}\right \vert \left
		\vert b_{r}^{\prime }\fancys b_{t}\right \vert \\
		&=&O_{p}\left( J^{2}n^{-2}\left( \sum_{s=t+1}^{n}\left \vert b_{t}^{\prime
		}\fancys b_{s}\right \vert \right) ^{2}\right) =O_{p}\left( J^{2}n^{-2}\left(
		\sum_{s=t+1}^{n}\sum_{i,j=1}^{n}\left \vert b_{it}\right \vert \left
		\vert s_{ij}\right \vert \left \vert b_{js}\right \vert \right)
		^{2}\right) ,
\end{eqnarray*}%
where we used Assumption \ref{ass:clt_general}. Now, by Assumption \ref{ass:clt_general},
\sloppy
$\sum_{s=t+1}^{n}\sum_{i,j=1}^{n}\left \vert b_{it}\right \vert \left
		\vert s_{ij}\right \vert \left \vert b_{js}\right \vert
		=O_p\left(\sup_{i,j}\left \vert s_{ij}\right \vert\sup_t\sum_{i=1}^{n}\left \vert b_{it}\right \vert \sum_{j=1}^{n}\sum_{s=t+1}^{n}\left \vert b_{js}\right
		\vert \right)=O_{p}\left( J\sup_t\sum_{i=1}^{n}\left \vert b_{it}\right \vert \right) ,$ and thus
 (\ref{pmle_var_11_2}) is
 $O_{p}\left( L^{-4}J^{4}n^{-2}\sup_{t}\left( \sum_{i=1}^{n}\left \vert b_{it}\right \vert \right) \left(
	\sum_{i=1}^{n}\left( \sum_{t=1}^{n-1}\left \vert b_{it}\right \vert
	\right) \right) \right)=O_{p}\left( L^{-4}J^{4}n^{-1}\right) $, whence the first term in (\ref{pmle_var_11_1}) is $O_{p}\left( J^2n^{-1}\right) $, which is negligible.
	
	Again transforming the range of summation and using the elementary inequality $|a+b|^2\leq C\left(a^2+b^2\right)$, we can bound the square of the second
	term in (\ref{pmle_var_11_1}) by
	\begin{equation}  \label{pmle_var_11_3}
		C\left(\sum_{t=1}^{n-1}\sum_{s=t+1}^n a_{st}^2
		\left( v_t^2-1\right)\right)^2+C\left(2\sum_{t=1}^{n-1}%
		\sum_{r=1}^{t-1}\sum_{s=t+1}^n a_{st} a_{sr}  v_t v_r\right)^2.
	\end{equation}
The expectations of the two terms in (\ref{pmle_var_11_3}) are bounded by a constant times $e _{1}$ and $e _{2}$, respectively, where
	$
	e _{1} =\sum_{t=1}^{n-1}\left( \sum_{s=t+1}^{n}a_{st}^{2}\right) ^{2},
	e _{2} =\sum_{t=1}^{n-1}\sum_{r=1}^{t-1}\left(
	\sum_{s=t+1}^{n}a_{st}a_{sr}\right) ^{2}.  $
	Thus (\ref{pmle_var_11_3}) is $O_{p}\left( e _{1}+e _{2}\right) $.
	Using Assumption \ref{ass:clt_general} and elementary inequalities, $e _{2}$ is bounded by
	\begin{eqnarray*}
		&&\sum_{t=1}^{n-1}\sum_{r=1}^{t-1}\sum_{s=t+1}^{n}\sum_{u=t+1}^{n}b_{s}^{%
			\prime }\fancys b_{t}b_{s}^{\prime }\fancys b_{r}b_{u}^{\prime }\fancys b_{t}b_{u}^{\prime }\fancys b_{r}
		\\
		&=&O_{p}\left( L^{-4}\sum_{s,r,t,u=1}^{n}\sum_{i,j=1}^{n}\left \vert
		b_{ir} \right \vert \left \vert s_{ij}\right \vert \left \vert
		b_{js} \right \vert \sum_{i,j=1}^{n}\left \vert b_{ir} \right
		\vert \left \vert s_{ij}\right \vert \left \vert b_{ju} \right \vert
		\sum_{i,j=1}^{n}\left \vert b_{it} \right \vert \left \vert
		s_{ij}\right \vert \left \vert b_{js} \right \vert
        \times
        \sum_{i,j=1}^{n}\left \vert b_{it} \right \vert \left \vert
		s_{ij}\right \vert \left \vert b_{ju} \right \vert \right) \\
		&=&O_{p}\left( L^{-4}pn^{-1}\sum_{s,r,t=1}^{n}\left(
		\sum_{i,j=1}^{n}\left \vert b_{ir} \right \vert \left \vert
		s_{ij}\right \vert \left \vert b_{js} \right \vert \right) \left(
		\sum_{i,j=1}^{n}\left \vert b_{ir} \right \vert \left \vert
		s_{ij}\right \vert \sum_{u=1}^{n}\left \vert b_{ju} \right \vert
		\right) \right. \\
		&\times &\left. \sum_{i,j=1}^{n}\left \vert b_{it} \right \vert \left
		\vert s_{ij}\right \vert \left \vert b_{js} \right \vert
		\sum_{i=1}^{n}\left \vert b_{it} \right \vert\sup_u \sum_{j=1}^{n}\left
		\vert b_{ju} \right \vert \right) \\
		&=&O_{p}\left( L^{-4}J^{2}n^{-2}\sum_{s,r=1}^{n}\left( \sum_{i,j=1}^{n}\left
		\vert b_{ir} \right \vert \left \vert s_{ij}\right \vert \left \vert
		b_{js} \right \vert \right) \sum_{i=1}^{n}\left \vert b_{ir}\right \vert\sum_{j=1}^{n}\left(\sum_{u=1}^n\left
		\vert b_{ju} \right \vert\right)\right.\\
		&\times&\left.\left( \sum_{i,j=1}^{n}\sum_{t=1}^{n}\left \vert b_{it}\right \vert \left \vert s_{ij}\right \vert \left \vert b_{js}\right \vert \right) \right) \\
		&=&O_{p}\left( L^{-4}J^{2}n^{-1}\sum_{i,j=1}^{n}\left( \sum_{r=1}^{n}\left
		\vert b_{ir} \right \vert \right) \left \vert s_{ij}\right \vert
		\left( \sum_{s=1}^{n}\left \vert b_{js} \right \vert \right)
		\left(\sup_{j} \sum_{i=1}^{n}\left \vert s_{ij}\right \vert \right)
		\sum_{j=1}^{n}\left \vert b_{js} \right \vert \right) \\
		&=&O_p\left(L^{-4}J^{2}n^{-1}\sup_{k}\sum_{i,j=1}^{n}\left \vert s_{ij}\right \vert \sum_{i=1}^{n}\left \vert s_{ik}\right \vert\right)=O_p\left(L^{-4}J^{2}n^{-1}\sup_{k}\sum_{i,j,\ell=1}^{n}\left \vert s_{ij}\right \vert \left \vert s_{\ell k}\right \vert\right)\\
		&=&O_p\left(L^{-4}J^{2}n^{-1}\sup_{k}\sum_{i,j,\ell=1}^{n}\left(s_{ij}^2+ s_{\ell k}^2\right)\right)=O_p\left(L^{-4}J^{2}n^{-1}\sum_{i,j,\ell=1}^{n}\left(s_{ij}^2+ s_{\ell j}^2\right)\right)\\&=&O_p\left(L^{-4}J^{2}n^{-1}\sum_{i,j=1}^{n}s_{ij}^2\right)=O_p\left(L^{-4}J^{2}\sup_{j}\sum_{i=1}^{n}s_{ij}^2\right)=O_p\left(Jn^{-1}\right),
	\end{eqnarray*}
	Similary $e
	_{1}$ is \allowdisplaybreaks%
	\begin{eqnarray*}
		&&O_{p}\left( L^{-4}\sum_{t=1}^{n-1}\left \{ \sum_{s=t+1}^{n}\left(
		\sum_{i,j=1}^{n}\left \vert s_{ij}\right \vert \left \vert b_{jt}\right \vert \left \vert b_{is}\right \vert \right) ^{2}\right \}
		^{2}\right) \\
		&=&O_{p}\left( L^{-4}\left(\sup_{i,j}\left\vert s_{ij}\right\vert\right)^4\sum_{t=1}^{n-1}\left
		\{ \sum_{s=t+1}^{n}\left(\sum_{i=1}^{n}\left\vert b_{is}\right\vert \sum_{j=1}^{n}\left \vert
		b_{jt}\right \vert \right) ^{2}\right \} ^{2}\right) \\
		&=&O_{p}\left( (LJN)^{-4}\sum_{t=1}^{n-1}\left
		\{ \sum_{s=t+1}^{n}\left(\sum_{i=1}^{n}\left\vert b_{is}\right\vert\right)^2\left( \sum_{j=1}^{n}\left \vert
		b_{jt}\right \vert \right) ^{2}\right \} ^{2}\right)\\
		&=&O_{p}\left( (LJN)^{-4}\sum_{t=1}^{n-1}\left( \sum_{s=t+1}^{n}\left(\sum_{i=1}^{n}\left\vert b_{is}\right\vert\right)^2\right)^2\left( \sum_{j=1}^{n}\left \vert
		b_{jt}\right \vert \right) ^{4}\right)\\
		&=&O_{p}\left( (LJN)^{-4}\left( \sum_{t=1}^{n-1}\sum_{j=1}^{n}\left \vert
		b_{jt}\right \vert \right) \left(\sum_{s=t+1}^{n}\sum_{i=1}^{n}\left\vert b_{is}\right\vert\right)^2
        \times
        \sup_s\left( \sum_{i=1}^{n}\left \vert
		b_{is}\right \vert \right) ^{2}\sup_t\left( \sum_{j=1}^{n}\left \vert
		b_{jt}\right \vert \right) ^{3}\right)\\
		&=&O_{p}\left( (L^{-4}J^{4}n^{-1}\right)=O_{p}\left( J^{2}n^{-1}\right).
	\end{eqnarray*}
\end{proof}

\section{Proofs of Theorems}\label{sec:appteststats}
\begin{proof}[Proof of Theorem \ref{theorem:W1appr}]
Denote $\rdelta=(r_{1\delta},\cdots,r_{n\delta})'$. Note that $\hat{\xi}=\xi+(L'\PKa L)^{-1}L'\PKa u=(L'\PKa L)^{-1}L'\PKa(\rdelta+\epsilon)$ and $\hat\alpha=\alpha+R(L'\PKa L)^{-1}L'\PKa(\rdelta+\epsilon)$. Under $H_0: \alpha=0$, we have $\hat\alpha=R(L'\PKa L)^{-1}L'\PKa(\rdelta+\epsilon)$. Then we are going to investigate
\begin{flalign}
	\begin{split}
		&n(\rdelta+\epsilon)'\PKa L(L'\PKa L)^{-1}R'\cD_{1}^{-1}R(L'\PKa L)^{-1}L'\PKa(\rdelta+\epsilon)\\
		=&\epsilon'\cM_1\epsilon+(A_1+A_2)
	\end{split}
\end{flalign}
where $\cM_1=\frac{1}{n}\Ka\cV_{1}\Ka'$ with $\cV_{1}=(\Ka'\Ka)^{-1}\Ka'L\Ta^{-1}R'\cD_{1}^{-1}R\Ta^{-1}L'\Ka(\Ka'\Ka)^{-1}$, $\Ta=\invn L'{\PKa}L$ and $\cD_{1}=\invn R\Ta^{-1}R'$, and $A_1=\rdelta'\cM_1\rdelta$ and $A_2=2\rdelta'\cM_1\epsilon$. Now, 
$|A_{1}|\leq \Vert r_{\delta} \Vert^2\Vert\cM_{1}\Vert=o_p(1)$, since $E\Vert\rdelta\Vert^2=o(1)$, and $\Vert\cM_{1}\Vert=O_p(1)$ by Assumptions \ref{ass:approxerrorsec2} and \ref{ass:eigsec2}. 
A similar argument holds for $A_2$ because it has (conditionally) zero mean and variance $4r_{\delta}'\cM_{1}E(\epsilon\epsilon')\cM_{1}r_{\delta}=4r_{\delta}'\cM_{1}r_{\delta}$, by idempotence of $\cM_{1}$. Thus $\frac{A_1}{\sqrt{d_\alpha}}$ and $\frac{A_2}{\sqrt{d_\alpha}}$ are both $o_p(1)$ and establishing the theorem.
\end{proof}
\begin{proof}[Proof of Theorem \ref{theorem:sar}]
This is a special case of Theorem \ref{theorem:sarlp} and therefore omitted.	
\end{proof}

\begin{proof}[Proof of Lemma \ref{lemma:shac1}]
This is a simpler version of the proof of Lemma \ref{lemma:shac3} below and is therefore omitted.
\end{proof}
\newcommand{\cN}{\mathscr{M}}
\begin{proof}[Proof of Theorem \ref{theorem:W2appr}]
$\hat\alpha=\alpha+R\hat\xi=\alpha+R(L'{\PKa}L)^{-1}L'{\PKa}(\rdelta+Bv)$ and under $H_0: \alpha=0$,
\begin{flalign}
	\begin{split}		n\hat\alpha'\widehat\cD_{2}^{-1}\hat\alpha=(v'\widehat\cM_2v+A_{11}+A_{12}+2A_{21}+2A_{22}),
	\end{split}
\end{flalign}
\sloppy where $A_{11}=\rdelta'(\widehat\cN_2-\cN_{2})\rdelta$, $A_{12}=\rdelta'\cN_{2}\rdelta$, $A_{21}=\rdelta'(\widehat\cM_{2}-\cM_{2})v$, $A_{22}=\rdelta'\cM_{2}v$, $\widehat\cM_2=\frac{1}{n}B'\Ka\widehat\cV_2\Ka'B$ with $\widehat\cV_2=(\Ka'\Ka)^{-1}\Ka'L\Ta^{-1}
R'\widehat\cD_{2}^{-1}R\Ta^{-1}
L'\Ka(\Ka'\Ka)^{-1}$, $\widehat\cD_2=\frac{1}{n}R\widehat\cU_1R'$, $\widehat\cU_1=\Ta^{-1}L'\Ka(\Ka'\Ka)^{-1}\hat\Xi_1(\Ka'\Ka)^{-1}\Ka'L\Ta^{-1}$, and 
$\widehat\cN_{2}=\frac{1}{n}\Ka\widehat\cV_2\Ka'$.
The remainder of the proof is similar to the proof of Theorem \ref{theorem:W4appr}.
\end{proof}
\begin{proof}[Proof of Theorem \ref{theorem:sarlp}]
The proof follows by setting $J=J_{1}$, $\fancya=\mathcal{M}_2$ and $\fancys=\invn\Ka\cV_{2}\Ka'$ in Theorem \ref{thm:general_clt}. Checking Assumption \ref{ass:clt_general} follows exactly as in the proof of Theorem \ref{theorem:sarvc} below and is omitted.	
\end{proof}
\begin{proof}[Proof of Lemma \ref{lemma:shac2}]
	This proof is a simpler version of the proof of Lemma \ref{lemma:shac3} below and is therefore omitted.
\end{proof}
\begin{proof}[Proof of Theorem \ref{theorem:W3appr}]
Define $\rG=(r_{1G},\cdots,r_{nG})'$ and $\rw=\rG+\rdelta$. Then, we have $\hat\alpha=\alpha+R\hat\xi=\alpha+R(F'{\PKb}F)^{-1}F'{\PKb}(\rw+Bv)$ and $\norm\rw\norm^2=o_p(1)$ by Assumption \ref{ass:approxerrorsec4}. Under $H_0: \alpha=0$,
\begin{flalign}
	\begin{split}		n\hat\alpha'\widehat\cD_{3}^{-1}\hat\alpha=(v'\widehat\cM_{3}v+A_{11}+A_{12}+2A_{21}+2A_{22}),
	\end{split}
\end{flalign}
where $A_{11}=\rw'(\widehat\cN_3-\cN_3)\rw, A_{12}=\rw'\cN_3\rw$, $A_{21}=\rw'(\widehat\cM_{3}-\cM_{3})v$, $A_{22}=\rw'\cM_{3}v$, $\widehat\cM_{3}=\frac{1}{n}B'\Kb\widehat\cV_{3}\Kb'B$ with $\widehat\cV_{3}=(\Kb'\Kb)^{-1}\Kb'F\Tb^{-1}R'\widehat\cD_{3}^{-1}R\Tb^{-1}F'\Kb(\Kb'\Kb)^{-1}$, $\widehat\cD_{3}=\frac{1}{n}R\widehat\cU_{2}R'$, 
$\widehat\cU_{2}=\Tb^{-1}F'\Kb(\Kb'\Kb)^{-1}\allowbreak
\hat\Xi_{2}(\Kb'\Kb)^{-1}\Kb'F\Tb^{-1}$, and $\widehat\cN_{3}=\frac{1}{n}\Kb\widehat\cV_{3}\Kb'$. The remainder of the proof is similar to the proof of Theorem \ref{theorem:W4appr}.
\end{proof}
\begin{proof}[Proof of Theorem \ref{theorem:sarlpnp}]
The proof is similar to that of Theorem \ref{theorem:sarvc} and omitted.	
\end{proof}
\begin{proof}[Proof of Lemma \ref{lemma:shac3}]
Define $r_{\lambda}=(r_{1\lambda},...,r_{n\lambda})'$.
First, note that
\begin{equation}\label{gammahatgamma}
	\hat\gamma-\gamma=(G'{\PKc}G)^{-1}G'\Kc(\Kc'\Kc)^{-1}\Kc'u=(G'{\PKc}G)^{-1}G'\Kc(\Kc'\Kc)^{-1}\Kc'(r_{\lambda}+r_{\delta}+\epsilon),
\end{equation}
and 
\begin{equation}\label{bound1}
\begin{split}
\E\left\Vert n^{-1} \Kc'(\rlambda+\rdelta)\right\Vert^2&\leq n^{-2}\left\Vert \Kc\right\Vert^2 \E\left\Vert \rlambda+\rdelta\right\Vert^2\\&
\leq
Cn^{-1}\overline{\alpha}(n^{-1}\Kc'\Kc)\left(\E\left\Vert r_{\lambda}\right\Vert^2+\E\left\Vert r_{\delta}\right\Vert^2\right)=o(1),
\end{split}
\end{equation}
because $\E(r_{i\lambda}^2)=o(n^{-1})$ and $\E(r_{i\delta}^2)=o(n^{-1})$, by Assumptions \ref{ass:approxerrorsec5} and \ref{ass:eigsec5}. 

Next,
\begin{eqnarray}
	\E\left\Vert n^{-1} \Kc'\epsilon\right\Vert^2&\leq& n^{-2}\sum_{i=1}^n\left\Vert k_i\right\Vert^2 \E \epsilon_i^2+n^{-2}\sum_{i\neq j}\left\Vert k_i\right\Vert\left\Vert k_j\right\Vert \left\vert \E \epsilon_i\epsilon_j\right\vert\nonumber\\
	&=& n^{-2}\sum_{i=1}^n\left\Vert k_i\right\Vert^2 \sum_{r}b_{ir}^2+2n^{-2}\sum_{j}\left\Vert k_j\right\Vert\sum_{i<j}\left\Vert k_i\right\Vert\sup_r\left\vert b_{jr}\right\vert \left(\sum_{r}\left\vert b_{ir}\right\vert\right)\nonumber\\
	&\leq& Cn^{-2}\sum_{i=1}^n\left\Vert k_i\right\Vert^2+Cn^{-2}\sup_{j}\left\Vert k_j\right\Vert\left(\sup_r\sum_{j}\left\vert b_{jr}\right\vert\right)\sum_{i<j}\left\Vert k_i\right\Vert\nonumber \\
	&=&O(\Jc/n), \label{bound2}
\end{eqnarray}
by Assumptions \ref{ass:bsums} and \ref{ass:secondmomentsec5}. 

Combining Assumption \ref{ass:eigsec5}, (\ref{gammahatgamma}), (\ref{bound1}) and (\ref{bound2}),
we obtain
\begin{equation}\label{gammabound}
	\left\Vert\hat\gamma-\gamma\right\Vert=O_p(\sqrt{\dgamma/n}).	
\end{equation}  
Now, observe that $\hat\Xi_{r,s}-\Xi_{r,s}=a_{1,rs}+a_{2,rs}+a_{3,rs},$ where
\begin{eqnarray}
	a_{1,rs}&=&n^{-1}\sum_{i,j=1}^nk_{ir}k_{sj}(\hat u_i\hat u_j-\epsilon_i\epsilon_j)\fancyk\left(\min_m \left\{d^*_{ij,m}/d_m\right\}\right),\label{a1def}\\
	a_{2,rs}&=&n^{-1}\sum_{i,j=1}^nk_{ir}k_{sj}(\epsilon_i\epsilon_j-\sigma_{ij})\fancyk\left(\min_m \left\{d^*_{ij,m}/d_m\right\}\right),\label{a2def}\\
	a_{3,rs}&=&n^{-1}\sum_{i,j=1}^nk_{ir}k_{sj}\sigma_{ij}\left(\fancyk\left(\min_m \left\{d^*_{ij,m}/d_m\right\}\right)-1\right),\label{a3def}
\end{eqnarray}
where recall that $\sigma_{ij}$ is a typical element of $\Sigma$. We will show $a_{l,rs}=o_p(1), l=1,2,3$. To begin, write $r_i=r_{i\lambda}+r_{i\delta}$, $g_i'$ for the $i$-th row of $G$ and note that
\begin{equation}\label{uhatuhatepseps}
	\hat u_i\hat u_j-\epsilon_i\epsilon_j=g_i'(\hat\gamma-\gamma)g_j'(\hat\gamma-\gamma)+g_i'(\hat\gamma-\gamma)r_j+r_ir_j+g_j'(\hat\gamma-\gamma)\epsilon_i+g_i'(\hat\gamma-\gamma)\epsilon_j+\epsilon_i r_j+\epsilon_j r_i.
\end{equation}

Now, by Assumptions \ref{ass:secondmomentsec5} and \ref{ass:approxerrorsec5}, $\left\vert r_i\right\vert=o_p(1/\sqrt{n})$, $\left\Vert g_i\right\Vert=O_p(\sqrt{d_\gamma})$ and, as in the proof of Theorem 2 of \citet{Kelejian2007}, $\sum_{j=1}^n \left(1-\prod_{m=1}^M\mathbf{1}(d^*_{ij,m}>d_m)\right)\leq \ell$, all uniformly in $i$. Then, we can write $\left\vert a_{1,rs}\right\vert\leq\sum_{l=1}^7a_{1l,rs}$,
where, writing $\pi_M=1-\prod_{m=1}^M\mathbf{1}(d^*_{ij,m}>d_m)$ and also using (\ref{gammabound}), 
\begin{eqnarray*}
	a_{11,rs}&=&n^{-1}\sum_{i,j=1}^n \pi_M\left\Vert g_i\right\Vert\left\Vert g_j\right\Vert\left\Vert \hat\gamma-\gamma\right\Vert^2\leq C \ell \cdot O_p\left({\dgamma^2}/{n}\right),\\
	a_{12,rs}&=&n^{-1}\sum_{i,j=1}^n \pi_M\left\Vert g_i\right\Vert\left\Vert \hat\gamma-\gamma\right\Vert\left\vert r_j\right\vert\leq C \ell\cdot O_p\left({ \dgamma}/{n}\right),\\
	a_{13,rs}&=&n^{-1}\sum_{i,j=1}^n \pi_M\left\vert r_i\right\vert\left\vert r_j\right\vert\leq C\ell\cdot O_p\left({1}/{n}\right),\\
	a_{14,rs}&=&n^{-1}\sum_{i,j=1}^n \pi_M\left\Vert g_j\right\Vert\left\Vert \hat\gamma-\gamma\right\Vert\left\vert \epsilon_i\right\vert\leq C\ell^{1-1/q}\cdot O_p\left({\dgamma}/{n^{1/2-1/q}}\right),\\
	a_{15,rs}&=&n^{-1}\sum_{i,j=1}^n \pi_M\left\Vert g_i\right\Vert\left\Vert \hat\gamma-\gamma\right\Vert\left\vert \epsilon_j\right\vert\leq C\ell^{1-1/q}\cdot O_p\left({\dgamma}/{n^{1/2-1/q}}\right),\\
	a_{16,rs}&=&n^{-1}\sum_{i,j=1}^n \pi_M\left\vert \epsilon_i\right\vert\left\vert r_j\right\vert\leq C\ell^{1-1/q}\cdot O_p\left({1}/{n^{1/2-1/q}}\right),\\
	a_{17,rs}&=&n^{-1}\sum_{i,j=1}^n \pi_M\left\vert \epsilon_j\right\vert\left\vert r_i\right\vert\leq C\ell^{1-1/q}\cdot O_p\left({1}/{n^{1/2-1/q}}\right),
\end{eqnarray*}
where the H\"older inequality is used to bound $a_{1j,rs}, j\geq 4$, as in the proof of Theorem 2 of \citet{Kelejian2007}. 
Therefore, we need $\frac{\ell \dgamma^2}{n}=o_p(1)$ and $\frac{\ell^{1-1/q}\dgamma}{n^{1/2-1/q}}=o_p(1)$ as $n\rightarrow\infty$.
By Assumption \ref{ass:distances}(a), we have $\frac{\ell \dgamma^2}{n}=o_p\left(n^{\eta-1}\dgamma^2\right)$ and $\frac{\ell^{1-1/q}\dgamma}{n^{1/2-1/q}}=o_p\left(n^{\eta(q-1)/q-(q-2)/2q}\dgamma\right)$.
We then have
$a_{1,rs}=o_p(1)$ under the sufficient rate in \eqref{rcV3} as $n\rightarrow\infty$.
$a_{2,rs}=o_p(1)$ and $a_{3,rs}=o_p(1)$ are established exactly as the negligibility of $b_{rs,n}$ and $c_{rs,n}$ is established in the proof of Theorem 1 of \citet{Kelejian2007}. This concludes the proof.
\end{proof}
\begin{proof}[Proof of Theorem \ref{theorem:W4appr}]
Denote $\rlambda=(r_{1\lambda},\cdots,r_{n\lambda})'$, $\ry=\rlambda+\rdelta$. Thus,
$\hat\vartheta=\vartheta+R\hat\xi=\vartheta+R(G'{\PKc}G)^{-1}G'{\PKc}(\ry+Bv)$ and $\norm\ry\norm^2=o(1)$ by Assumption \ref{ass:approxerrorsec5}. Under $H_0: \vartheta=0$,
\begin{flalign}
	\begin{split}		n\hat\vartheta'\widehat\cD_{4}^{-1}\hat\vartheta=(v'\widehat\cM_4v+A_{11}+A_{12}+2A_{21}+2A_{22})
	\end{split}
\end{flalign}
where $A_{11}=\ry'(\widehat\cN_4-\cN_4)\ry, A_{12}=\ry'\cN_4\ry$, $A_{21}=\ry'(\widehat\cM_{4}-\cM_{4})v$, $A_{22}=\ry'\cM_{4}v$, $\widehat\cM_{4}=\frac{1}{n}B'\Kc\widehat\cV_4\Kc'B$, $\widehat\cV_4=(\Kc'\Kc)^{-1}\Kc'G\Tc^{-1}R'\widehat\cD_{4}^{-1}R\Tc^{-1}G'\Kc(\Kc'\Kc)^{-1}$, $\widehat\cD_4=\frac{1}{n}R\widehat\cU_3R'$, 
$\widehat\cU_3=\Tc^{-1}G'\Kc(\Kc'\Kc)^{-1}\allowbreak
\hat\Xi_3(\Kc'\Kc)^{-1}\Kc'G\Tc^{-1}$, and $\widehat\cN_{4}=\frac{1}{n}\Kc\widehat\cV_4\Kc'$. 
Observe that
\sloppy
$\cU_3=\Tc^{-1}\left(\invn L'\Kc\right)\left(\invn \Kc'\Kc\right)^{-1}\Xi_{3}\left(\invn \Kc'\Kc\right)^{-1}\left(\invn \Kc'L\right)\Tc^{-1}$,
we have $\norm\widehat{\cU}_3-\cU_3\norm\leq Cr_n\cdot O_p\left(1\right)=o_p(1)$ according to the proof of Lemma \ref{lemma:shac3}
and Assumptions \ref{ass:secondmomentsec5} and \ref{ass:eigsec5}.
Thus, $\norm\widehat{\cU}_{3}^{-1}-\cU_3^{-1}\norm\leq\norm\widehat{\cU}_{3}^{-1}\norm\norm\cU_3-\widehat{\cU}_{3}\norm\norm\cU_{3}^{-1}\norm\leq Cr_n\cdot O_p\left(1\right)=o_p(1)$ by Assumptions \ref{ass:eigsec3} and \ref{ass:eigsec5}, where $r_n=\frac{\ell\dgamma^2}{n}+\frac{\ell^{1-1/q}\dgamma}{n^{1/2-1/q}}$.
Similarly, we have
$\norm\widehat\cN_{4}-\cN_{4}\norm=o_p(1)$, and $\norm\widehat{\cM}_4-\cM_4\norm=o_p(1)$ by the structure of $\cN_{4}$, $\cM_4$ and Assumptions \ref{ass:bsums} and \ref{ass:secondmomentsec5}.
\par 
Since $\E(v_{i}^2)=1$, we have
\begin{flalign}\label{approrate}
\frac{1}{\sqrt{\dvartheta}}\big|v'(\widehat{\cM}_4-\cM_4)v\big|\leq C r_n\cdot O_p\left(\dvartheta^{-1/2}\right)=o_p(1) 
\end{flalign}
\par 
Similar to \eqref{approrate} and the proof of Theorem \ref{theorem:sar}, we have $\frac{1}{\sqrt{\dvartheta}}\left\vert A_{11}\right\vert\leq Cr_n\cdot O_p\left(1\right)=o_p(1)$, 
$\frac{1}{\sqrt{\dvartheta}}|A_{12}|\leq Cr_n\cdot O_p\left(1\right)=o_p(1)$ since $\cN_{4}=O_p(1)$ by Assumptions \ref{ass:bsums} and \ref{ass:secondmomentsec5},
$\frac{1}{\sqrt{\dvartheta}}|A_{21}|\leq Cr_n\cdot O_p\left(1\right)=o_p(1)$
and
$\frac{1}{\sqrt{\dvartheta}}|A_{22}|=o_p(1)$.
Therefore, we obtain the sufficient rate from \eqref{rcV3} as $n\rightarrow\infty$ for the required quantities to be negligible.
\end{proof}
\begin{proof}[Proof of Theorem \ref{theorem:sarvc}]	
(i) Set $J=J_{3}$, $\fancya=\mathcal{M}_4$ and $\fancys=\invn\Kc\cV_4\Kc'$ in Theorem \ref{thm:general_clt}. 
We now note that the entries of $\fancys$ are $s_{ij,4}=\invn k_i'\cV_4 k_j,$
where $k_i'$ is the $i$-th row of $\Kc$, and so $\left \vert s_{ij,4}\right \vert =O_{p}\left( n^{-1}\left\Vert k_i\right\Vert\left\Vert k_j\right\Vert\right)= O_{p}\left( \Jc n^{-1}\right),$ uniformly in $i,j$, by Assumptions \ref{ass:secondmomentsec5} and \ref{ass:eigsec5}, and  arguments used in the proof of Theorem \ref{theorem:W4appr} above.
Similarly, we also observe that 
$\sum_{j=1}^n s_{ij,4}^2=O_p\left(\Jc n^{-1}\right),$ uniformly in $i$. Thus Assumption \ref{ass:clt_general}(b) is satisfied. Assumption \ref{ass:clt_general} is satisfied to due to Assumption \ref{ass:eigsec5}. The conclusion now follows by Theorem \ref{thm:general_clt}.\par 

(ii) By Lemma \ref{lemma:shac3}, we have 
\begin{flalign}\label{eq:A1}
\left\Vert\hat\vartheta-\vartheta\right\Vert=O_p\left(\sqrt{\dvartheta/n}\right).
\end{flalign}
From Assumptions \ref{ass:secondmomentsec5} and \ref{ass:eigsec5}, we have
$\left\Vert\widehat{\cD}_4-\cD_4\right\Vert=o_p(1) $
and 
\begin{flalign}\label{eq:A2}
\left\Vert\widehat{\cD}_{4}^{-1}-\cD_4^{-1}\right\Vert\leq\left\Vert\widehat{\cD}_{4}^{-1}\right\Vert\left\Vert\cD_4-\widehat{\cD}_{4}\right\Vert\left\Vert\cD_{4}^{-1}\right\Vert=o_p(1).
\end{flalign}
Write $\hat\vartheta=\vartheta+\Delta_\vartheta$ from \eqref{eq:A1} where $\Delta_\vartheta=O_p\left(\sqrt{\dvartheta/n}\right)$. Then
$
n\hat\vartheta'\hat{\cD}_4^{-1}\hat\vartheta
=n\vartheta'\cD_4^{-1}\vartheta+ 2n\vartheta'\cD_4^{-1}\Delta_\vartheta+n\Delta_\vartheta'\cD_4^{-1}\Delta_\vartheta+nR_n,
$ where 
$\left\vert R_n\right\vert=\left\vert\hat\vartheta'\hat{\cD}_4^{-1}\hat\vartheta-\hat\vartheta'\cD_4^{-1}\hat\vartheta\right\vert=o_p(1)\cdot\left\Vert\hat\vartheta\right\Vert^2=o_p\left(\Vert\vartheta\Vert^2+\dvartheta/n\right)=o_p\left(\Vert\vartheta\Vert^2\right)$
by \eqref{eq:A2}.
\par 
Using \eqref{rc5.2}, \eqref{eq:A1} and bounded eigenvalues in \eqref{eq:A2},
\begin{equation}\label{eq:A4}
\begin{split}
&\frac{1}{\sqrt{2\dvartheta}}2n\vartheta'\cD_4^{-1}\Delta_\vartheta
=O\left(\frac{n}{\sqrt{\dvartheta}}\right)O\left(\Vert\vartheta\Vert\right)O_p(1)O_p\left(\sqrt{\dvartheta/n}\right)
=o_p\left(\frac{n\Vert\vartheta\Vert}{\sqrt{\dvartheta}}\right)
\\
&\frac{1}{\sqrt{2\dvartheta}}n\Delta_\vartheta'\cD_4^{-1}\Delta_\vartheta
=O\left(\frac{n}{\sqrt{\dvartheta}}\right)O_p(1)O_p\left(\frac{\dvartheta}{n}\right)
=o_p\left(\frac{n}{\sqrt{\dvartheta}}\right)\\
&\frac{1}{\sqrt{2\dvartheta}}nR_n
=o_p\left(\frac{n\Vert\vartheta\Vert^2}{\sqrt{\dvartheta}}\right).
\end{split}
\end{equation}
Note that for the sparse case, $\Vert\vartheta\Vert=O(1)$, and for the dense case, $\Vert\vartheta\Vert_{\infty}=O(1)$ and $\Vert\vartheta\Vert=\sqrt{\dvartheta}$. 
Thus, in general, $\mathbb{W}_{4}=\frac{n\vartheta'\cD_{4}^{-1}\vartheta-\dvartheta}{\sqrt{2\dvartheta}}+o_p\Big(\frac{n\Vert\vartheta\Vert^2}{\sqrt{\dvartheta}}\Big)$.  For any nonstochastic sequence $\{C_n\}, C_n=o\left(\frac{n\Vert\vartheta\Vert^2}{\sqrt{\dvartheta}}\right), \lim_{n\to\infty}P\left(\left|\mathbb{W}_4\right|>C_n\right)=1$, so that consistency follows.
\par  
\newcommand{\He}{H_{\ell}}
(iii) We cover all local alternatives here.  
We define $H_{\ell i}\equiv H_{\ell i,n}$, with corresponding statistic $\mathbb{W}_{i}^{*}$, for $i=1,2,3,4$. For $\mathbb{W}_{1}^{*}$, we consider $\alpha=\alpha^{*}\equiv\nu_{1n}\da^{\frac{1}{4}}/(n\nu_{1n}'\Gamma_{1n}\nu_{1n})^{\frac{1}{2}}$ with $\nu_{1n}$ a $\da\times 1$ non-zero vector and $\Gamma_{1n}$ a $\da\times\da$ matrix. Under $H_{\ell 1}$, we have $\mathbb{W}_{1}^{*}=\mathbb{W}_{1}+\frac{1}{\sqrt{2\da}}\da^{\frac{1}{2}}\mathcal{A}_{1}\left(\frac{1}{n}\nu_{1n}'\Gamma_{1n}\nu_{1n}\right)^{-1}+o_p(1),$ where $\mathcal{A}_{1}=\frac{1}{n}\nu_{1n}'\cD_{1}^{-1}\nu_{1n}$. We choose $\Gamma_{1n}=\cD_{1}^{-1}$ and then finish the proof.
For $\mathbb{W}_{4}^{*}$, we consider
$\vartheta=\vartheta^{*}\equiv\nu_{4n}\dvartheta^{\frac{1}{4}}/\left(n\nu_{4n}'\Gamma_{4n}\nu_{4n}\right)^{\frac{1}{2}}$. Under $H_{\ell 4}$, we have 
$\mathbb{W}_{4}^{*}=\mathbb{W}_{4}+\frac{1}{\sqrt{2\dvartheta}}\dvartheta^{\frac{1}{2}}\mathcal{A}_{4}\left(\frac{1}{n}\nu_{4n}'\Gamma_{4n}\nu_{4n}\right)^{-1}+o_p(1),$ 
where $\mathcal{A}_{4}=\frac{1}{n}\nu_{4n}'\widehat\cD_{4}^{-1}\nu_{4n}$,
and $\left\Vert\frac{1}{n}\nu_{4n}'(\widehat\cD_{4}^{-1}-\cD_{4}^{-1})\nu_{4n}\right\Vert=o_p(1)$ by \eqref{rcV3}.
We choose $\Gamma_{4n}=\widehat\cD_{4}^{-1}$ and then finish the proof. The proofs of 
$\mathbb{W}_{2}^{*}$ and $\mathbb{W}_{3}^{*}$ 
\end{proof}

\setcounter{table}{0}
\renewcommand{\thetable}{S\arabic{table}}
\setcounter{figure}{0}
\renewcommand{\thefigure}{S\arabic{figure}}

\section{Supplementary materials}

This section presents additional simulation results. Additional results for the polynomial functions in the main text are shown in Tables~\ref{tab:mcW4size_lambda}-\ref{tab:mcW4power_lambda}, and those for the trigonometric functions in the main text are shown in Tables~\ref{tab:mcW1size_trig}-\ref{tab:mcW4power_trig_lambda}.

We also considered tests of the CRS hypothesis to facilitate the empirical part in Section \ref{empcrs}, with $d_{\delta}=2$. Our null of interest when testing CRS is
$
H_0^{true}: \delta_1(z)+\delta_2(z)=1, z\in\mathcal{Z}.
$
We reparameterized with $\theta(z)=\delta_1(z)+\delta_2(z)-1$ and the model is
$y_i-p_{i1}=\lambda_1 \suml_{j=1}^{n}w_{1ij}y_j+x_i^{\prime}\beta+p_{i1}\theta(z_i) +(p_{i2}-p_{i1})\delta_{2}(z_i)+\epsilon_{i}.$ 
For the alternatives, we set $\delta_{1}(\rx)=1-\rx^2$ and $\delta_{2}(\rx)=1+\rx^2$. We studied all four statistics proposed in the main text. These results are similar to our main results, and so the corresponding tables are omitted and available upon request.

\begin{table}[H]
	\centering\footnotesize\stl{2.3mm}\stla{0mm}
\caption{Empirical sizes of $\mathbb{W}_{4}$ for the varying spatial coefficient SAR model proposed in Section~\ref{sec:vcsar}. 
Null approximating $H_0^{true}$, $H_0:\mu=0$. 
Rejection probabilities at 1\%, 5\% and 10\% levels. Polynomial basis functions.}
	\begin{spacing}{0.9}
\begin{tabular}{lllrrrrrrrrr}
    \toprule\toprule
    \multicolumn{3}{c}{\multirow{2}[2]{*}{$\dlambda=1$}} & \multicolumn{3}{c}{$n=200$} & \multicolumn{3}{c}{$n=500$} & \multicolumn{3}{c}{$n=900$} \\
\cmidrule{4-12}    \multicolumn{3}{c}{}  & \multicolumn{1}{c}{1\%} & \multicolumn{1}{c}{5\%} & \multicolumn{1}{c}{10\%} & \multicolumn{1}{c}{1\%} & \multicolumn{1}{c}{5\%} & \multicolumn{1}{c}{10\%} & \multicolumn{1}{c}{1\%} & \multicolumn{1}{c}{5\%} & \multicolumn{1}{c}{10\%} \\
    \midrule\midrule
    \multirow{6}[2]{*}{\ref{V1}} & \multirow{2}[1]{*}{$h=2$} & asy-p & 0.081  & 0.109  & 0.126  & 0.044  & 0.072  & 0.102  & 0.018  & 0.044  & 0.063  \\
          &       & chi-p & 0.050  & 0.100  & 0.122  & 0.012  & 0.063  & 0.097  & 0.005  & 0.029  & 0.059  \\
          & \multirow{2}[1]{*}{$h=4$} & asy-p & 0.100  & 0.146  & 0.171  & 0.052  & 0.081  & 0.103  & 0.011  & 0.027  & 0.041  \\
          &       & chi-p & 0.072  & 0.128  & 0.169  & 0.039  & 0.068  & 0.100  & 0.005  & 0.039  & 0.091  \\
          & \multirow{2}[1]{*}{$h=8$} & asy-p & 0.129  & 0.162  & 0.198  & 0.087  & 0.126  & 0.159  & 0.058  & 0.085  & 0.122  \\
          &       & chi-p & 0.100  & 0.150  & 0.190  & 0.068  & 0.116  & 0.153  & 0.020  & 0.066  & 0.106  \\
\end{tabular}%
	\end{spacing}
	\label{tab:mcW4size_lambda}%
\end{table}%

\begin{table}[H]
	\centering\footnotesize\stl{2.3mm}\stla{0mm}
\caption{Empirical powers of $\mathbb{W}_{2}$ for the baseline SAR model with SHAC correction proposed in Section \ref{sec:spaterror}. 
Null approximating $H_0^{true}$, $H_{0}:\alpha=0$. 
Rejection probabilities at 1\%, 5\% and 10\% levels. Polynomial basis functions.
P-values: asy-p (standard normal) and chi-p (chi-square calibration $(\chi^2_{\da}-\da)/\sqrt{2\da}$).}
	\begin{spacing}{1}
    \begin{tabular}{lllrrrrrrrrr}
    \toprule\toprule
    \multicolumn{3}{c}{\multirow{2}[2]{*}{$\dlambda=2$}} & \multicolumn{3}{c}{$n=200$} & \multicolumn{3}{c}{$n=500$} & \multicolumn{3}{c}{$n=900$} \\
\cmidrule{4-12}    \multicolumn{3}{c}{}  & \multicolumn{1}{c}{1\%} & \multicolumn{1}{c}{5\%} & \multicolumn{1}{c}{10\%} & \multicolumn{1}{c}{1\%} & \multicolumn{1}{c}{5\%} & \multicolumn{1}{c}{10\%} & \multicolumn{1}{c}{1\%} & \multicolumn{1}{c}{5\%} & \multicolumn{1}{c}{10\%} \\
    \midrule\midrule
    \multirow{6}[2]{*}{\ref{V1}} & \multirow{2}[1]{*}{$h=2$} & asy-p & 0.979  & 0.989  & 0.991  & 0.996  & 0.997  & 0.997  & 1  & 1  & 1  \\
          &       & chi-p & 0.969  & 0.984  & 0.991  & 0.995  & 0.996  & 0.997  & 0.999  & 1  & 1  \\
          & \multirow{2}[1]{*}{$h=4$} & asy-p & 0.997  & 0.999  & 0.999  & 0.999  & 0.999  & 0.999  & 1  & 1  & 1  \\
          &       & chi-p & 0.995  & 0.997  & 0.999  & 0.999  & 0.999  & 0.999  & 1  & 1  & 1  \\
          & \multirow{2}[1]{*}{$h=8$} & asy-p & 0.920  & 0.921  & 0.923  & 0.987  & 0.987  & 0.987  & 1  & 1  & 1  \\
          &       & chi-p & 0.920  & 0.920  & 0.923  & 0.987  & 0.987  & 0.987  & 1  & 1  & 1  \\
\end{tabular}%
	\end{spacing}
	\label{tab:mcW2power}%
\end{table}%

\begin{table}[H]
	\centering\footnotesize\stl{2.3mm}\stla{0mm}
\caption{Empirical powers of $\mathbb{W}_{3}$ for the SAR model with nonparametric spatial weights proposed in Section \ref{sec:sarlpnp}. 
Null approximating $H_0^{true}$, $H_{0}:\alpha=0$. 
Rejection probabilities at 1\%, 5\% and 10\% levels. Polynomial basis functions.
P-values: asy-p (standard normal) and chi-p (chi-square calibration $(\chi^2_{\da}-\da)/\sqrt{2\da}$).}
	\begin{spacing}{1}
    \begin{tabular}{lllrrrrrrrrr}
    \toprule\toprule
    \multicolumn{3}{c}{\multirow{2}[3]{*}{$\dtau=2$}} & \multicolumn{3}{c}{$n=200$} & \multicolumn{3}{c}{$n=500$} & \multicolumn{3}{c}{$n=900$} \\
\cmidrule{4-12}    \multicolumn{3}{c}{}  & \multicolumn{1}{c}{1\%} & \multicolumn{1}{c}{5\%} & \multicolumn{1}{c}{10\%} & \multicolumn{1}{c}{1\%} & \multicolumn{1}{c}{5\%} & \multicolumn{1}{c}{10\%} & \multicolumn{1}{c}{1\%} & \multicolumn{1}{c}{5\%} & \multicolumn{1}{c}{10\%} \\
    \midrule\midrule
    \multirow{6}[2]{*}{\ref{V1}} & \multirow{2}[1]{*}{$h=2$} & asy-p & 0.979  & 0.989  & 0.991  & 0.996  & 0.997  & 0.997  & 1  & 1  & 1  \\
          &       & chi-p & 0.969  & 0.984  & 0.991  & 0.995  & 0.996  & 0.997  & 0.999  & 1  & 1  \\
          & \multirow{2}[1]{*}{$h=4$} & asy-p & 0.997  & 0.999  & 0.999  & 0.999  & 0.999  & 0.999  & 1  & 1  & 1  \\
          &       & chi-p & 0.995  & 0.997  & 0.999  & 0.999  & 0.999  & 0.999  & 1  & 1  & 1  \\
          & \multirow{2}[1]{*}{$h=8$} & asy-p & 0.920  & 0.921  & 0.923  & 0.987  & 0.987  & 0.987  & 1  & 1  & 1  \\
          &       & chi-p & 0.920  & 0.920  & 0.923  & 0.987  & 0.987  & 0.987  & 1  & 1  & 1  \\
\end{tabular}%
	\end{spacing}
	\label{tab:mcW3power}%
\end{table}%

\begin{table}[H]
	\centering\footnotesize\stl{2.3mm}\stla{0mm}
\caption{Empirical powers of $\mathbb{W}_{4}$ for the varying spatial coefficient SAR model proposed in Section \ref{sec:vcsar}. 
Null approximating $H_0^{true}$, $H_{0}:\alpha=0$. 
Rejection probabilities at 1\%, 5\% and 10\% levels. Polynomial basis functions.
P-values: asy-p (standard normal) and chi-p (chi-square calibration $(\chi^2_{\da}-\da)/\sqrt{2\da}$).}
	\begin{spacing}{0.9}
    \begin{tabular}{lllrrrrrrrrr}
    \toprule\toprule
    \multicolumn{3}{c}{\multirow{2}[2]{*}{$\dlambda=1$}} & \multicolumn{3}{c}{$n=200$} & \multicolumn{3}{c}{$n=500$} & \multicolumn{3}{c}{$n=900$} \\
\cmidrule{4-12}    \multicolumn{3}{c}{}  & \multicolumn{1}{c}{1\%} & \multicolumn{1}{c}{5\%} & \multicolumn{1}{c}{10\%} & \multicolumn{1}{c}{1\%} & \multicolumn{1}{c}{5\%} & \multicolumn{1}{c}{10\%} & \multicolumn{1}{c}{1\%} & \multicolumn{1}{c}{5\%} & \multicolumn{1}{c}{10\%} \\
    \midrule\midrule
    \multirow{6}[1]{*}{\ref{V1}} & \multirow{2}[1]{*}{$h=2$} & asy-p & 0.972  & 0.977  & 0.984  & 0.990  & 0.991  & 0.996  & 0.996  & 0.997  & 0.998  \\
          &       & chi-p & 0.977  & 0.980  & 0.984  & 0.991  & 0.993  & 0.996  & 0.997  & 0.997  & 0.998  \\
          & \multirow{2}[1]{*}{$h=4$} & asy-p & 0.952  & 0.969  & 0.981  & 0.973  & 0.983  & 0.989  & 0.991  & 0.997  & 0.999  \\
          &       & chi-p & 0.963  & 0.976  & 0.982  & 0.980  & 0.986  & 0.990  & 0.995  & 0.998  & 0.999  \\
          & \multirow{2}[0]{*}{$h=8$} & asy-p & 0.901  & 0.924  & 0.945  & 0.948  & 0.962  & 0.968  & 0.972  & 0.980  & 0.986  \\
          &       & chi-p & 0.915  & 0.931  & 0.947  & 0.955  & 0.966  & 0.971  & 0.978  & 0.982  & 0.986  \\
\end{tabular}%
	\end{spacing}
	\label{tab:mcW4power}%
\end{table}%

\begin{table}[H]
	\centering\footnotesize\stl{2.3mm}\stla{0mm}
\caption{Empirical powers of $\mathbb{W}_{4}$ for the varying spatial coefficient SAR model proposed in Section \ref{sec:vcsar}. 
Null approximating $H_0^{true}$, $H_{0}:\mu=0$. 
Rejection probabilities at 1\%, 5\% and 10\% levels. Polynomial basis functions.
P-values: asy-p (standard normal) and chi-p (chi-square calibration $(\chi^2_{d_{\mu}}-d_{\mu})/\sqrt{2d_{\mu}}$).}
	\begin{spacing}{0.9}
    \begin{tabular}{lllrrrrrrrrr}
    \toprule\toprule
    \multicolumn{3}{c}{\multirow{2}[2]{*}{$\dlambda=1$}} & \multicolumn{3}{c}{$n=200$} & \multicolumn{3}{c}{$n=500$} & \multicolumn{3}{c}{$n=900$} \\
\cmidrule{4-12}    \multicolumn{3}{c}{}  & \multicolumn{1}{c}{1\%} & \multicolumn{1}{c}{5\%} & \multicolumn{1}{c}{10\%} & \multicolumn{1}{c}{1\%} & \multicolumn{1}{c}{5\%} & \multicolumn{1}{c}{10\%} & \multicolumn{1}{c}{1\%} & \multicolumn{1}{c}{5\%} & \multicolumn{1}{c}{10\%} \\
    \midrule\midrule
    \multirow{6}[2]{*}{\ref{V1}} & \multirow{2}[1]{*}{$h=2$} & asy-p & 0.974  & 0.980  & 0.980  & 0.980  & 0.989  & 0.989  & 0.988  & 0.992  & 0.992  \\
          &       & chi-p & 0.998  & 0.998  & 0.998  & 0.998  & 0.999  & 0.999  & 0.999  & 1  & 1  \\
          & \multirow{2}[1]{*}{$h=4$} & asy-p & 0.923  & 0.958  & 0.966  & 0.973  & 0.990  & 0.990  & 0.988  & 0.992  & 0.992  \\
          &       & chi-p & 0.946  & 0.976  & 0.984  & 0.991  & 1  & 1  & 0.999  & 1  & 1  \\
          & \multirow{2}[1]{*}{$h=8$} & asy-p & 0.785  & 0.847  & 0.876  & 0.791  & 0.849  & 0.878  & 0.810  & 0.869  & 0.898  \\
          &       & chi-p & 0.805  & 0.863  & 0.893  & 0.806  & 0.858  & 0.887  & 0.820  & 0.877  & 0.906  \\
\end{tabular}%
	\end{spacing}
	\label{tab:mcW4power_lambda}%
\end{table}%

\begin{table}[H]
\centering\footnotesize\stl{2.3mm}\stla{0mm}
\caption{Empirical sizes of $\mathbb{W}_{1}$ for the baseline SAR model proposed in Section \ref{sec:sar}. 
Null approximating $H_0^{true}$, $H_{0}:\alpha=0$. 
Rejection probabilities at 1\%, 5\% and 10\% levels. Trigonometric basis functions.
P-values: asy-p (standard normal) and chi-p (chi-square calibration $(\chi^2_{\da}-\da)/\sqrt{2\da}$).
}
	\begin{spacing}{1}
    \begin{tabular}{lllrrrrrrrrr}
    \toprule\toprule
    \multicolumn{3}{c}{\multirow{2}[2]{*}{$\dlambda=2$}} & \multicolumn{3}{c}{$n=200$} & \multicolumn{3}{c}{$n=500$} & \multicolumn{3}{c}{$n=900$} \\
\cmidrule{4-12}    \multicolumn{3}{c}{}  & \multicolumn{1}{c}{1\%} & \multicolumn{1}{c}{5\%} & \multicolumn{1}{c}{10\%} & \multicolumn{1}{c}{1\%} & \multicolumn{1}{c}{5\%} & \multicolumn{1}{c}{10\%} & \multicolumn{1}{c}{1\%} & \multicolumn{1}{c}{5\%} & \multicolumn{1}{c}{10\%} \\
    \midrule\midrule
    \multirow{6}[2]{*}{\ref{V1}} & \multirow{2}[1]{*}{$h=2$} & asy-p & 0.034  & 0.075  & 0.107  & 0.029  & 0.068  & 0.100  & 0.037  & 0.077  & 0.106  \\
          &       & chi-p & 0.014  & 0.050  & 0.103  & 0.010  & 0.045  & 0.099  & 0.011  & 0.055  & 0.103  \\
          & \multirow{2}[1]{*}{$h=4$} & asy-p & 0.032  & 0.085  & 0.112  & 0.026  & 0.069  & 0.112  & 0.038  & 0.085  & 0.118  \\
          &       & chi-p & 0.017  & 0.060  & 0.108  & 0.009  & 0.049  & 0.104  & 0.014  & 0.062  & 0.113  \\
          & \multirow{2}[1]{*}{$h=8$} & asy-p & 0.028  & 0.052  & 0.085  & 0.024  & 0.050  & 0.077  & 0.034  & 0.084  & 0.126  \\
          &       & chi-p & 0.014  & 0.046  & 0.078  & 0.007  & 0.038  & 0.072  & 0.012  & 0.064  & 0.117  \\
\end{tabular}%
	\end{spacing}
	\label{tab:mcW1size_trig}%
\end{table}%

\begin{table}[H]
	\centering\footnotesize\stl{2.3mm}\stla{0mm}
\caption{Empirical sizes of $\mathbb{W}_{2}$ for the baseline SAR model with SHAC correction proposed in Section \ref{sec:spaterror}. 
Null approximating $H_0^{true}$, $H_{0}:\alpha=0$. 
Rejection probabilities at 1\%, 5\% and 10\% levels. Trigonometric basis functions.
P-values: asy-p (standard normal) and chi-p (chi-square calibration $(\chi^2_{\da}-\da)/\sqrt{2\da}$).}
	\begin{spacing}{1}
    \begin{tabular}{lllrrrrrrrrr}
    \toprule\toprule
    \multicolumn{3}{c}{\multirow{2}[2]{*}{$\dlambda=2$}} & \multicolumn{3}{c}{$n=200$} & \multicolumn{3}{c}{$n=500$} & \multicolumn{3}{c}{$n=900$} \\
\cmidrule{4-12}    \multicolumn{3}{c}{}  & \multicolumn{1}{c}{1\%} & \multicolumn{1}{c}{5\%} & \multicolumn{1}{c}{10\%} & \multicolumn{1}{c}{1\%} & \multicolumn{1}{c}{5\%} & \multicolumn{1}{c}{10\%} & \multicolumn{1}{c}{1\%} & \multicolumn{1}{c}{5\%} & \multicolumn{1}{c}{10\%} \\
    \midrule\midrule
    \multirow{6}[2]{*}{\ref{V1}} & \multirow{2}[1]{*}{$h=2$} & asy-p & 0.038  & 0.062  & 0.080  & 0.026  & 0.061  & 0.094  & 0.042  & 0.082  & 0.118  \\
          &       & chi-p & 0.010  & 0.050  & 0.078  & 0.009  & 0.040  & 0.090  & 0.013  & 0.058  & 0.115  \\
          & \multirow{2}[1]{*}{$h=4$} & asy-p & 0.051  & 0.081  & 0.107  & 0.034  & 0.063  & 0.092  & 0.040  & 0.082  & 0.114  \\
          &       & chi-p & 0.027  & 0.062  & 0.103  & 0.013  & 0.047  & 0.086  & 0.013  & 0.065  & 0.110  \\
          & \multirow{2}[1]{*}{$h=8$} & asy-p & 0.082  & 0.120  & 0.146  & 0.034  & 0.066  & 0.081  & 0.032  & 0.070  & 0.111  \\
          &       & chi-p & 0.051  & 0.104  & 0.141  & 0.016  & 0.049  & 0.078  & 0.021  & 0.054  & 0.103  \\
\end{tabular}%
	\end{spacing}
	\label{tab:mcW2size_trig}%
\end{table}%

\begin{table}[H]
	\centering\footnotesize\stl{2.3mm}\stla{0mm}
\caption{Empirical sizes of $\mathbb{W}_{3}$ for the SAR model with nonparametric spatial weights proposed in Section \ref{sec:sarlpnp}. 
Null approximating $H_0^{true}$, $H_{0}:\alpha=0$. 
Rejection probabilities at 1\%, 5\% and 10\% levels. Trigonometric basis functions.
P-values: asy-p (standard normal) and chi-p (chi-square calibration $(\chi^2_{\da}-\da)/\sqrt{2\da}$).}
	\begin{spacing}{1}
    \begin{tabular}{lllrrrrrrrrr}
    \toprule\toprule
    \multicolumn{3}{c}{\multirow{2}[2]{*}{$\dtau=2$}} & \multicolumn{3}{c}{$n=200$} & \multicolumn{3}{c}{$n=500$} & \multicolumn{3}{c}{$n=900$} \\
\cmidrule{4-12}    \multicolumn{3}{c}{}  & \multicolumn{1}{c}{1\%} & \multicolumn{1}{c}{5\%} & \multicolumn{1}{c}{10\%} & \multicolumn{1}{c}{1\%} & \multicolumn{1}{c}{5\%} & \multicolumn{1}{c}{10\%} & \multicolumn{1}{c}{1\%} & \multicolumn{1}{c}{5\%} & \multicolumn{1}{c}{10\%} \\
    \midrule\midrule
    \multirow{6}[1]{*}{\ref{V1}} & \multirow{2}[1]{*}{$h=2$} & asy-p & 0.042  & 0.067  & 0.087  & 0.029  & 0.054  & 0.086  & 0.036  & 0.062  & 0.080  \\
          &       & chi-p & 0.018  & 0.060  & 0.084  & 0.007  & 0.042  & 0.084  & 0.012  & 0.042  & 0.080  \\
          & \multirow{2}[1]{*}{$h=4$} & asy-p & 0.054  & 0.090  & 0.119  & 0.029  & 0.067  & 0.091  & 0.028  & 0.061  & 0.089  \\
          &       & chi-p & 0.024  & 0.071  & 0.115  & 0.011  & 0.046  & 0.085  & 0.011  & 0.046  & 0.089  \\
          & \multirow{2}[0]{*}{$h=8$} & asy-p & 0.060  & 0.117  & 0.147  & 0.040  & 0.073  & 0.104  & 0.034  & 0.065  & 0.092  \\
          &       & chi-p & 0.012  & 0.095  & 0.143  & 0.019  & 0.059  & 0.101  & 0.016  & 0.056  & 0.089  \\
\end{tabular}%
	\end{spacing}
	\label{tab:mcW3size_trig}%
\end{table}%

\begin{table}[H]
	\centering\footnotesize\stl{2.3mm}\stla{0mm}
\caption{Empirical sizes of $\mathbb{W}_{4}$ for the varying spatial coefficient SAR model proposed in Section \ref{sec:vcsar}. 
Null approximating $H_0^{true}$, $H_{0}:\alpha=0$. 
Rejection probabilities at 1\%, 5\% and 10\% levels. Trigonometric basis functions.
P-values: asy-p (standard normal) and chi-p (chi-square calibration $(\chi^2_{\da}-\da)/\sqrt{2\da}$).}
	\begin{spacing}{0.9}
    \begin{tabular}{lllrrrrrrrrr}
    \toprule\toprule
    \multicolumn{3}{c}{\multirow{2}[2]{*}{$\dlambda=1$}} & \multicolumn{3}{c}{$n=200$} & \multicolumn{3}{c}{$n=500$} & \multicolumn{3}{c}{$n=900$} \\
\cmidrule{4-12}    \multicolumn{3}{c}{}  & \multicolumn{1}{c}{1\%} & \multicolumn{1}{c}{5\%} & \multicolumn{1}{c}{10\%} & \multicolumn{1}{c}{1\%} & \multicolumn{1}{c}{5\%} & \multicolumn{1}{c}{10\%} & \multicolumn{1}{c}{1\%} & \multicolumn{1}{c}{5\%} & \multicolumn{1}{c}{10\%} \\
    \midrule\midrule
    \multirow{6}[2]{*}{\ref{V1}} & \multirow{2}[1]{*}{$h=2$} & asy-p & 0.057  & 0.076  & 0.097  & 0.042  & 0.056  & 0.078  & 0.044  & 0.068  & 0.085  \\
          &       & chi-p & 0.033  & 0.063  & 0.096  & 0.024  & 0.050  & 0.075  & 0.019  & 0.056  & 0.084  \\
          & \multirow{2}[1]{*}{$h=4$} & asy-p & 0.088  & 0.126  & 0.147  & 0.054  & 0.086  & 0.100  & 0.041  & 0.072  & 0.092  \\
          &       & chi-p & 0.059  & 0.106  & 0.143  & 0.034  & 0.068  & 0.097  & 0.027  & 0.062  & 0.088  \\
          & \multirow{2}[1]{*}{$h=8$} & asy-p & 0.085  & 0.126  & 0.143  & 0.029  & 0.045  & 0.059  & 0.023  & 0.030  & 0.074  \\
          &       & chi-p & 0.067  & 0.109  & 0.139  & 0.018  & 0.038  & 0.058  & 0.015  & 0.038  & 0.083  \\
\end{tabular}%
	\end{spacing}
	\label{tab:mcW4size_trig}%
\end{table}%

\begin{table}[H]
	\centering\footnotesize\stl{2.3mm}\stla{0mm}
\caption{Empirical sizes of $\mathbb{W}_{4}$ for the varying spatial coefficient SAR model proposed in Section \ref{sec:vcsar}. 
Null approximating $H_0^{true}$, $H_{0}:\mu=0$. 
Rejection probabilities at 1\%, 5\% and 10\% levels. Trigonometric basis functions.
P-values: asy-p (standard normal) and chi-p (chi-square calibration $(\chi^2_{d_{\mu}}-d_{\mu})/\sqrt{2d_{\mu}}$).}
	\begin{spacing}{0.9}
    \begin{tabular}{lllrrrrrrrrr}
    \toprule\toprule
    \multicolumn{3}{c}{\multirow{2}[2]{*}{$\dlambda=1$}} & \multicolumn{3}{c}{$n=200$} & \multicolumn{3}{c}{$n=500$} & \multicolumn{3}{c}{$n=900$} \\
\cmidrule{4-12}    \multicolumn{3}{c}{}  & \multicolumn{1}{c}{1\%} & \multicolumn{1}{c}{5\%} & \multicolumn{1}{c}{10\%} & \multicolumn{1}{c}{1\%} & \multicolumn{1}{c}{5\%} & \multicolumn{1}{c}{10\%} & \multicolumn{1}{c}{1\%} & \multicolumn{1}{c}{5\%} & \multicolumn{1}{c}{10\%} \\
    \midrule\midrule
    \multirow{6}[2]{*}{\ref{V1}} & \multirow{2}[1]{*}{$h=2$} & asy-p & 0.059  & 0.097  & 0.131  & 0.020  & 0.049  & 0.076  & 0.014  & 0.028  & 0.037  \\
          &       & chi-p & 0.029  & 0.082  & 0.122  & 0.006  & 0.031  & 0.072  & 0.002  & 0.039  & 0.080  \\
          & \multirow{2}[1]{*}{$h=4$} & asy-p & 0.097  & 0.136  & 0.165  & 0.019  & 0.027  & 0.040  & 0.001  & 0.004  & 0.011  \\
          &       & chi-p & 0.061  & 0.121  & 0.158  & 0.010  & 0.033  & 0.078  & 0.001  & 0.036  & 0.080  \\
          & \multirow{2}[1]{*}{$h=8$} & asy-p & 0.074  & 0.120  & 0.141  & 0.054  & 0.061  & 0.169  & 0.017  & 0.037  & 0.076  \\
          &       & chi-p & 0.048  & 0.103  & 0.139  & 0.023  & 0.025  & 0.070  & 0.014  & 0.034  & 0.073  \\
\end{tabular}%
	\end{spacing}
	\label{tab:mcW4size_lambda_trig}%
\end{table}%

\begin{table}[H]
	\centering\footnotesize\stl{2.3mm}\stla{0mm}
\caption{Empirical powers of $\mathbb{W}_{2}$ for the baseline SAR model with SHAC correction proposed in Section \ref{sec:spaterror}. 
Null approximating $H_0^{true}$, $H_{0}:\alpha=0$. 
Rejection probabilities at 1\%, 5\% and 10\% levels. Trigonometric basis functions.
P-values: asy-p (standard normal) and chi-p (chi-square calibration $(\chi^2_{\da}-\da)/\sqrt{2\da}$).}
	\begin{spacing}{1}
    \begin{tabular}{lllrrrrrrrrr}
    \toprule\toprule
    \multicolumn{3}{c}{\multirow{2}[2]{*}{$\dlambda=2$}} & \multicolumn{3}{c}{$n=200$} & \multicolumn{3}{c}{$n=500$} & \multicolumn{3}{c}{$n=900$} \\
\cmidrule{4-12}    \multicolumn{3}{c}{}  & \multicolumn{1}{c}{1\%} & \multicolumn{1}{c}{5\%} & \multicolumn{1}{c}{10\%} & \multicolumn{1}{c}{1\%} & \multicolumn{1}{c}{5\%} & \multicolumn{1}{c}{10\%} & \multicolumn{1}{c}{1\%} & \multicolumn{1}{c}{5\%} & \multicolumn{1}{c}{10\%} \\
    \midrule\midrule
    \multirow{6}[2]{*}{\ref{V1}} & \multirow{2}[1]{*}{$h=2$} & asy-p & 0.995  & 0.997  & 0.999  & 0.997  & 0.998  & 0.998  & 1  & 1  & 1  \\
          &       & chi-p & 0.980  & 0.996  & 0.999  & 0.997  & 0.998  & 0.998  & 1  & 1  & 1  \\
          & \multirow{2}[1]{*}{$h=4$} & asy-p & 0.996  & 0.999  & 0.999  & 1  & 1  & 1  & 1  & 1  & 1  \\
          &       & chi-p & 0.996  & 0.997  & 0.999  & 1  & 1  & 1  & 1  & 1  & 1  \\
          & \multirow{2}[1]{*}{$h=8$} & asy-p & 0.978  & 0.979  & 0.979  & 0.999  & 0.999  & 0.999  & 1  & 1  & 1  \\
          &       & chi-p & 0.976  & 0.979  & 0.979  & 0.999  & 0.999  & 0.999  & 1  & 1  & 1  \\
\end{tabular}%
	\end{spacing}
	\label{tab:mcW2power_trig}%
\end{table}%

\begin{table}[H]
	\centering\footnotesize\stl{2.3mm}\stla{0mm}
\caption{Empirical powers of $\mathbb{W}_{3}$ for the SAR model with nonparametric spatial weights proposed in Section \ref{sec:sarlpnp}. 
Null approximating $H_0^{true}$, $H_{0}:\alpha=0$. 
Rejection probabilities at 1\%, 5\% and 10\% levels. Trigonometric basis functions.
P-values: asy-p (standard normal) and chi-p (chi-square calibration $(\chi^2_{\da}-\da)/\sqrt{2\da}$).}
	\begin{spacing}{1}
    \begin{tabular}{lllrrrrrrrrr}
    \toprule\toprule
    \multicolumn{3}{c}{\multirow{2}[3]{*}{$\dtau=2$}} & \multicolumn{3}{c}{$n=200$} & \multicolumn{3}{c}{$n=500$} & \multicolumn{3}{c}{$n=900$} \\
\cmidrule{4-12}    \multicolumn{3}{c}{}  & \multicolumn{1}{c}{1\%} & \multicolumn{1}{c}{5\%} & \multicolumn{1}{c}{10\%} & \multicolumn{1}{c}{1\%} & \multicolumn{1}{c}{5\%} & \multicolumn{1}{c}{10\%} & \multicolumn{1}{c}{1\%} & \multicolumn{1}{c}{5\%} & \multicolumn{1}{c}{10\%} \\
    \midrule\midrule
    \multirow{6}[1]{*}{\ref{V1}} & \multirow{2}[1]{*}{$h=2$} & asy-p & 0.974  & 0.983  & 0.988  & 0.973  & 0.983  & 0.989  & 0.982  & 0.992  & 0.994  \\
          &       & chi-p & 0.942  & 0.980  & 0.988  & 0.954  & 0.977  & 0.988  & 0.966  & 0.988  & 0.994  \\
          & \multirow{2}[1]{*}{$h=4$} & asy-p & 0.986  & 0.996  & 0.998  & 0.986  & 0.995  & 0.996  & 0.994  & 0.997  & 0.998  \\
          &       & chi-p & 0.976  & 0.994  & 0.998  & 0.973  & 0.991  & 0.996  & 0.985  & 0.997  & 0.998  \\
          & \multirow{2}[0]{*}{$h=8$} & asy-p & 0.853  & 0.857  & 0.859  & 0.894  & 0.901  & 0.903  & 0.996  & 0.997  & 0.998  \\
          &       & chi-p & 0.849  & 0.855  & 0.858  & 0.689  & 0.901  & 0.903  & 0.992  & 0.997  & 0.998  \\
\end{tabular}%
	\end{spacing}
	\label{tab:mcW3power_trig}%
\end{table}%

\begin{table}[H]
	\centering\footnotesize\stl{2.3mm}\stla{0mm}
\caption{Empirical powers of $\mathbb{W}_{4}$ for the varying spatial coefficient SAR model proposed in Section \ref{sec:vcsar}. 
Null approximating $H_0^{true}$, $H_{0}:\alpha=0$. 
Rejection probabilities at 1\%, 5\% and 10\% levels. Trigonometric basis functions.
P-values: asy-p (standard normal) and chi-p (chi-square calibration $(\chi^2_{\da}-\da)/\sqrt{2\da}$).}
	\begin{spacing}{0.9}
    \begin{tabular}{lllrrrrrrrrr}
    \toprule\toprule
    \multicolumn{3}{c}{\multirow{2}[2]{*}{$\dlambda=1$}} & \multicolumn{3}{c}{$n=200$} & \multicolumn{3}{c}{$n=500$} & \multicolumn{3}{c}{$n=900$} \\
\cmidrule{4-12}    \multicolumn{3}{c}{}  & \multicolumn{1}{c}{1\%} & \multicolumn{1}{c}{5\%} & \multicolumn{1}{c}{10\%} & \multicolumn{1}{c}{1\%} & \multicolumn{1}{c}{5\%} & \multicolumn{1}{c}{10\%} & \multicolumn{1}{c}{1\%} & \multicolumn{1}{c}{5\%} & \multicolumn{1}{c}{10\%} \\
    \midrule\midrule
    \multirow{6}[1]{*}{\ref{V1}} & \multirow{2}[1]{*}{$h=2$} & asy-p & 0.973  & 0.982  & 0.990  & 0.988  & 0.992  & 0.995  & 0.996  & 0.998  & 0.998  \\
          &       & chi-p & 0.977  & 0.986  & 0.990  & 0.992  & 0.995  & 0.995  & 0.997  & 0.998  & 0.998  \\
          & \multirow{2}[1]{*}{$h=4$} & asy-p & 0.947  & 0.968  & 0.975  & 0.975  & 0.982  & 0.985  & 0.993  & 0.996  & 0.997  \\
          &       & chi-p & 0.959  & 0.972  & 0.975  & 0.976  & 0.984  & 0.985  & 0.996  & 0.996  & 0.997  \\
          & \multirow{2}[0]{*}{$h=8$} & asy-p & 0.913  & 0.931  & 0.943  & 0.945  & 0.961  & 0.969  & 0.976  & 0.984  & 0.988  \\
          &       & chi-p & 0.922  & 0.935  & 0.943  & 0.956  & 0.963  & 0.969  & 0.980  & 0.985  & 0.988  \\
\end{tabular}%
	\end{spacing}
	\label{tab:mcW4power_trig}%
\end{table}%

\begin{table}[H]
	\centering\footnotesize\stl{2.3mm}\stla{0mm}
\caption{Empirical powers of $\mathbb{W}_{4}$ for the varying spatial coefficient SAR model proposed in Section \ref{sec:vcsar}. 
Null approximating $H_0^{true}$, $H_{0}:\mu=0$. 
Rejection probabilities at 1\%, 5\% and 10\% levels. Trigonometric basis functions.
P-values: asy-p (standard normal) and chi-p (chi-square calibration $(\chi^2_{d_{\mu}}-d_{\mu})/\sqrt{2d_{\mu}}$).}
	\begin{spacing}{0.9}
    \begin{tabular}{lllrrrrrrrrr}
    \toprule\toprule
    \multicolumn{3}{c}{\multirow{2}[2]{*}{$\dlambda=1$}} & \multicolumn{3}{c}{$n=200$} & \multicolumn{3}{c}{$n=500$} & \multicolumn{3}{c}{$n=900$} \\
\cmidrule{4-12}    \multicolumn{3}{c}{}  & \multicolumn{1}{c}{1\%} & \multicolumn{1}{c}{5\%} & \multicolumn{1}{c}{10\%} & \multicolumn{1}{c}{1\%} & \multicolumn{1}{c}{5\%} & \multicolumn{1}{c}{10\%} & \multicolumn{1}{c}{1\%} & \multicolumn{1}{c}{5\%} & \multicolumn{1}{c}{10\%} \\
    \midrule\midrule
    \multirow{6}[2]{*}{\ref{V1}} & \multirow{2}[1]{*}{$h=2$} & asy-p & 0.973  & 0.979  & 0.979  & 0.979  & 0.989  & 0.989  & 0.988  & 0.991  & 0.992  \\
          &       & chi-p & 0.997  & 0.997  & 0.997  & 0.997  & 0.999  & 0.999  & 0.999  & 0.999  & 1  \\
          & \multirow{2}[1]{*}{$h=4$} & asy-p & 0.969  & 0.978  & 0.979  & 0.982  & 0.990  & 0.990  & 0.989  & 0.992  & 0.992  \\
          &       & chi-p & 0.993  & 0.996  & 0.997  & 1  & 1  & 1  & 1  & 1  & 1  \\
          & \multirow{2}[1]{*}{$h=8$} & asy-p & 0.923  & 0.941  & 0.947  & 0.975  & 0.985  & 0.986  & 0.985  & 0.991  & 0.992  \\
          &       & chi-p & 0.946  & 0.959  & 0.965  & 0.993  & 0.995  & 0.996  & 0.996  & 0.999  & 1  \\
\end{tabular}%
	\end{spacing}
	\label{tab:mcW4power_trig_lambda}%
\end{table}%

\begin{figure}[H]
  \centering
  \hspace*{-1.5cm}
  \begin{minipage}[t]{0.6\textwidth}
    \centering
    \includegraphics[width=\linewidth]{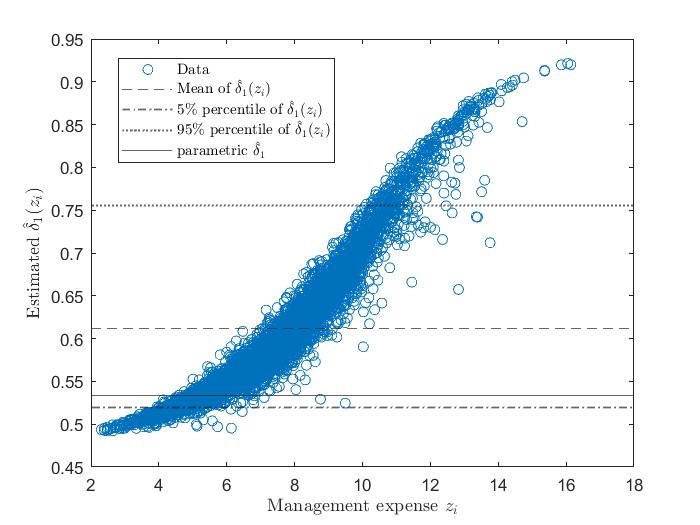}
\footnotesize{(a) Capital}
  \end{minipage}\hspace*{-0.6cm}
  \begin{minipage}[t]{0.6\textwidth}
    \centering
    \includegraphics[width=\linewidth]{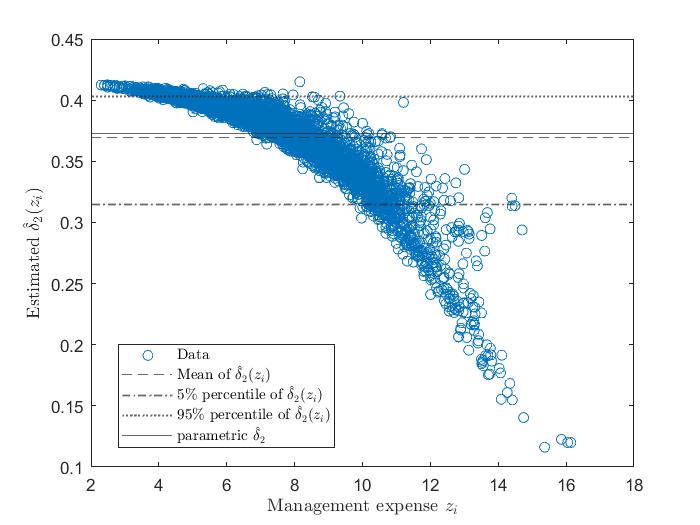}
 \footnotesize{(b) Labor}
\end{minipage}

  \begin{minipage}[t]{0.6\textwidth}
    \centering
    \includegraphics[width=\linewidth]{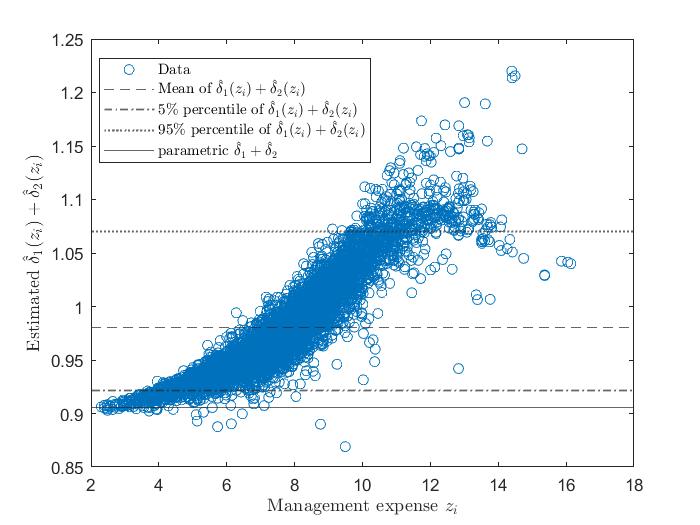}
  \footnotesize{(c) CRS}
\end{minipage}

  \caption{Cobb-Douglas versus semiparametric production.}
  \label{fig:crs_comp}
\end{figure}

\end{document}